\pdfoutput=1 
\documentclass[11pt,letterpaper]{article}
\newcommand{\ifarxiv}[2]{#2}
\renewcommand{\ifarxiv}[2]{#1}

\usepackage[utf8]{inputenc}
\ifarxiv{
\usepackage[margin=0.9in]{geometry}
}{
\usepackage[margin=1in]{geometry}
}
\usepackage{palatino}
\usepackage{microtype}

\usepackage{helvet}

\usepackage{amsmath}
\usepackage{amsfonts}
\usepackage{amssymb}
\usepackage{amsthm}
\usepackage{comment}
\usepackage{mathtools}
\usepackage{mathrsfs}
\usepackage{bm}
\usepackage{caption}
\usepackage{subcaption}
\usepackage{sectsty}
\allsectionsfont{\sffamily}

\usepackage[dvipsnames]{xcolor}
\usepackage{tikz}
\usepackage{mathpazo}
\usepackage{verbatim}
\usepackage{tcolorbox}

\usepackage[linesnumbered,ruled,lined]{algorithm2e}
\usepackage{algpseudocode}

\SetKwFor{ForPar}{for}{in parallel do}{end}
\SetKwFor{ForAllPar}{forall}{in parallel do}{end}

\usepackage{hyperref}
\usepackage[capitalize, nameinlink]{cleveref}

\newcommand{\todo}[1]{\typeout{TODO: \the\inputlineno: #1}\textbf{{\color{red}[[[ #1 ]]]}}}

\def\ID{{\mathfrak{I}}}

\def\dep{{\mathsf{depth}}}

\def\dtv{{d_{\mathrm{TV}}}}

\def\ID{{\mathcal{I}}}

\newtheorem{theorem}{Theorem}[section]

\newtheorem*{claim*}{Claim}

\newtheorem{lemma}[theorem]{Lemma}
\newtheorem{proposition}[theorem]{Proposition}
\newtheorem{corollary}[theorem]{Corollary}

\theoremstyle{definition}
\newtheorem{condition}{Condition}
\newtheorem{definition}[theorem]{Definition}
\newtheorem{remark}[theorem]{Remark}
\newtheorem*{remark*}{Remark}

\newtheorem*{case*}{Case}

\newcommand{\CurConfig}[3]{{#1}^{({#2})}_{{#3}}}

\newcommand{\Sample}{\mathsf{Sample}}

\newcommand{\Sampler}{\mathsf{ConsistSampler}}
\newcommand{\RandSeed}{\mathcal{R}}
\newcommand{\InfMat}{\boldsymbol{\rho}}
\newcommand{\SampleMat}{\mathcal{S}}
\newcommand{\Inf}[2]{\InfMat\left({#1},{#2}\right)}
\newcommand{\CONGEST}{$\mathsf{CONGEST}$}
\newcommand{\PRAM}{$\mathsf{PRAM}$}
\newcommand{\CRAM}{$\mathsf{CRCW}$-$\mathsf{PRAM}$}
\newcommand{\NC}{$\mathsf{NC}$}
\newcommand{\RNC}{$\mathsf{RNC}$}

\renewcommand{\prec}{\rightarrow}

\title{Parallelize Single-Site Dynamics up to Dobrushin Criterion}
\date{}
\author{Hongyang Liu~\thanks{State Key Laboratory for Novel Software Technology,  New Cornerstone Science Laboratory, Nanjing University,  China. E-mails: \url{liuhongyang@smail.nju.edu.cn}, \url{yinyt@nju.edu.cn}} 
\and Yitong Yin\footnotemark[1]
}

\begin{document}
\maketitle
\begin{abstract}
Single-site dynamics are canonical Markov chain based algorithms for sampling from high-dimensional distributions, such as the Gibbs distributions of graphical models.
We introduce a simple and generic parallel algorithm that faithfully simulates single-site dynamics.
Under a much relaxed, asymptotic variant of the $\ell_p$-Dobrushin's condition---where the Dobrushin's influence matrix has a bounded $\ell_p$-induced operator norm for an arbitrary $p\in[1, \infty]$---our algorithm simulates  $N$ steps of single-site updates within a parallel depth of $O\left({N}/{n}+\log n\right)$ on $\tilde{O}(m)$ processors, 
where $n$ is the number of sites and $m$ is the size of the graphical model.
For Boolean-valued random variables, if the $\ell_p$-Dobrushin's condition holds---specifically, if the $\ell_p$-induced operator norm of the Dobrushin's influence matrix is less than~$1$---the parallel depth can be further reduced to $O(\log N+\log n)$, achieving an exponential speedup.

These results suggest that single-site dynamics with near-linear mixing times can be parallelized into \RNC{} sampling algorithms, 
independent of the maximum degree of the underlying graphical model, 
as long as the Dobrushin influence matrix maintains a bounded operator norm. 
We show the effectiveness of this approach with \RNC{} samplers for the hardcore and Ising models within their uniqueness regimes, 
as well as an \RNC{} SAT sampler for satisfying solutions of CNF formulas in a local lemma regime. 
Furthermore, by employing non-adaptive simulated annealing, these \RNC{} samplers can be transformed into \RNC{} algorithms for approximate counting.

\end{abstract}
\setcounter{page}{0} 
\thispagestyle{empty} 
\vfill

\pagebreak


\section{Introduction}
Drawing random samples according to prescribed probability distributions is a fundamental computational problem. 
Historically, Monte Carlo simulations, which rely on random sampling, were among the earliest computer programs~\cite{metropolis1987beginning}. 
Today, sampling from high-dimensional distributions remains a central challenge across various fields of computer science and data science.

The Markov chain Monte Carlo  (MCMC) method is one of the primary methods for sampling.
A significant portion of Markov chains used for sampling from high-dimensional distributions belongs to the class of \emph{single-site dynamics}.
Let $V$ be a set of $n$ {sites}, and  $Q$ be a set of $q=|Q|$ {spins}.
Let $\mu$ be a distribution defined over all {configurations} $\sigma\in Q^V$.
A canonical Markov chain for sampling from $\mu$ is the following single-site dynamics, 
known as the \emph{Glauber dynamics} (also called the \emph{Gibbs sampler}  or \emph{heat-bath dynamics}):
\begin{itemize}
\item
Given the current configuration $\sigma \in Q^V$, a site $v \in V$ is picked uniformly at random, and its spin $\sigma_v$ is updated by drawing a new spin independently according to the marginal distribution $\mu_v^{\sigma_{V \setminus {v}}}$.
\end{itemize}
It is well known that the stationary distribution of this chain is $\mu$.
Furthermore, when $\mu$ is a {Gibbs distribution} of a graphical model defined by local constraints, 
there exists an underlying graph $G=(V,E)$ such that the marginal distributions $\mu_v^{\sigma_{V\setminus\{v\}}}$ depend only on $\sigma_{N_v}$, the values of $\sigma$ at the neighbors of $v$.

Abstractly, a single-site dynamics on a graph $G=(V,E)$ can be defined by a class of local update distributions $\{P_v^{\tau}\}$,
where each $P_v^\tau$ is a distribution over $Q$. 
This distribution is determined by the site $v\in V$  chosen for update and the current configuration $\tau=\sigma_{N_v^+}$ on $v$'s inclusive neighborhood $N_v^+$.\footnote{We use the inclusive neighborhood $N_v^+=\{u\mid\{u,v\}\in E\}\cup\{v\}$ in this abstraction of single-site dynamics because, in some single-site dynamics, the rule for updating a site $v$ may depend on the current spin of $v$ itself, e.g.~in Metropolis chains.}

\paragraph{Parallel MCMC sampling.}
MCMC sampling plays a pivotal role in  efficiently performing Monte Carlo calculations for a variety of complex tasks,
including volume estimation and integration~\cite{dyer1991random}, approximation of partition functions~\cite{vstefankovivc2009adaptive,jerrum1993polynomial},  permanent computation~\cite{jerrum2004polynomial}, and counting satisfying solutions~\cite{feng2021fast}.
These tasks are fundamental to statistical inference and data analysis, which have become major focuses in the era of big data.
As data scales, the demand for solving these tasks efficiently by leveraging parallel computing resources has grown significantly, 
making the parallelization of MCMC sampling increasingly critical. 
A substantial body of research, both in practice~\cite{jordan2019communication,niu2011hogwild,smola2010architecture,terenin2020asynchronous,sa2016ensuring, ahmed2012scalable, gonzalez2011parallel,aadit2022massively,scott2022bayes,daskalakis2018hogwild, neiswanger2013asymptotically,mahani2015simd} and in theory~\cite{carlson2023improved, AnariHLVXY23, AnariBTV23, feng2021distributed,feng2021dynamic,anari2021sampling,biswas2020local, guo2019uniform, feng2018local, fischer2018simple, daskalakis2018hogwild, feng2017sampling, sa2016ensuring, gonzalez2011parallel},
is dedicated to advancing parallel MCMC methods. 
This ongoing effort aims to address the challenge of efficiently conducting Monte Carlo calculations on massive datasets in contemporary computational environments.

The single-site dynamics are inherently defined in a sequential manner. 
This  raises intriguing theoretical questions about whether such sequentiality is intrinsic to the problems they solve.
In a seminal work on parallel computing, Mulmuley, Vazirani and Vazirani~\cite{mulmuley1987matching} asked whether an efficient parallel algorithm could exist for sampling bipartite perfect matchings, a key step towards approximating permanents in parallel.
However, Teng \cite{teng1995independent} later provided negative evidence to this question by identifying barriers to parallelizing Markov chains, particularly those related to perfect matchings.
Teng further conjectured that 0-1 permanents are not efficiently approximable in the \NC{} class,
making it a rare example of a problem believed to be intrinsically sequential but not known to be P-complete.
This conjecture, along with the barrier, reflects the belief that simulating a dynamical system may be inevitably sequential.

In fact, single-site dynamics originated as an idealized continuous-time parallel process~\cite{glauber1963time}.
\begin{center}
  \begin{tcolorbox}[=sharpish corners, colback=white, width=1\linewidth]
  	\begin{center}
	\textbf{\emph{The Continuous-Time Single-Site Dynamics}}
  	\end{center}
    The continuous-time process $(X_t)_{t\in\mathbb{R}_{\ge0}}$ on space $Q^V$ evolves as:
    \begin{itemize}
    \item each site $v\in V$ holds an independent rate-1 Poisson clock;
	\item when a clock at a site $v\in V$  rings, the spin $X_t(v)$ is redrawn  according to $P_v^{X_t(N_v^+)}$.
    \end{itemize}
  \end{tcolorbox} 
\end{center}

This continuous-time process effectively parallelizes the single-site dynamics, achieving a speedup factor of $n = |V|$. 
This is because a continuous time $T \in \mathbb{R}_{\ge0}$ corresponds to $N$ discrete steps, where $N \sim \text{Pois}(nT)$.

Although this idealized parallel process has been known for over half a century,
little is known about how to implement it correctly and efficiently on a parallel computer. 
A major obstacle is a classical conundrum in concurrency:
if two adjacent sites concurrently update their spins based on the current configurations of their neighborhoods, 
it can lead to an incorrect simulation of the original continuous-time process, where updates are atomic.
On the other hand, avoiding this issue by disallowing concurrent updates of adjacent sites incurs an extra time complexity factor of the degree $\Delta_G$ (or more precisely, the chromatic number $\chi_G$) of the graphical model $G$, since at most one site per neighborhood can be updated at any moment.
This obstacle is encountered by traditional parallel samplers based on chromatic scheduler~\cite{gonzalez2011parallel}.
For general graphical models $G$ with $n$ sites, 
the maximum degree $\Delta_G$ or the chromatic number $\chi_G$ can be as large as $\Omega(n)$, 
making it challenging to simulate MCMC sampling using \RNC{} programs within polylogarithmic depth of parallel computations.

We are therefore concerned with the following fundamental question:
Can we simulate continuous-time single-site dynamics both faithfully and efficiently on parallel computers?
The question concerns simulating an idealization of our physical world~\cite{glauber1963time} on man-made computing systems with minimal errors or overheads. 
This goal appears ambitious, especially considering the barriers identified in~\cite{teng1995independent}. 
Recently, some positive results have been shown for a subclass of single-site dynamics known as Metropolis chains~\cite{feng2021distributed}. 
However, this approach was specifically tailored to these chains, where updates can be reduced to coin flipping.
For general single-site dynamics, particularly Glauber dynamics, where updates involve more complex decisions beyond Boolean choices, it remains largely uncertain whether these processes can be efficiently and faithfully simulated in parallel.

\subsection{Our Results}
We present a simple and generic parallel algorithm that can faithfully simulate single-site dynamics.

\paragraph{Dobrushin's criterion.}
The parallel simulation is efficient under a condition formulated similarly to Dobrushin's conditions, but is significantly weaker in scale.
\begin{definition}[Dobrushin's influence matrix~\cite{dobrushin1970prescribing} ]\label{def:dobrushin-matrix}
The \emph{Dobrushin's influence matrix}  $\InfMat\in\mathbb{R}_{\ge 0}^{V\times V}$  for the single-site dynamics on the graph $G=(V,E)$ specified by $\{P_v^{\tau}\}$ is defined as:
\begin{align}
\forall u,v\in V:\quad  \Inf{u}{v}\triangleq \max_{\tau,\tau':\tau\oplus\tau'\subseteq\{u\}} \dtv\left(P_v^{\tau},P_v^{\tau'}\right),\label{eq:dobrushin-matrix}
\end{align}
where 
 $\dtv(\cdot,\cdot)$ denotes the total variation distance.
  The maximum is taken over all pairs of configurations  $\tau,\tau'\in Q^{N_v^+}$ on the inclusive neighborhood $N_v^+=N_v\cup\{v\}$ of $v$,  such that $\tau$ and $\tau'$ agree with each other everywhere except at $u$.
\end{definition}

For distinct non-adjacent $u,v\in V$, the maximum in~\eqref{eq:dobrushin-matrix} is actually taken over all $(\tau,\tau')$ that $\tau=\tau'$, and hence $\Inf{u}{v}=0$ for such pairs of $u,v$.
For instance, in the case of Glauber dynamics, it holds that $\Inf{v}{v}=0$; however, in general, the diagonal entries of $\rho$ can be non-zero, as seen in the Metropolis chain.

The $\ell_p$-induced operator norm of Dobrushin's influence matrix is defined by: 
\[
\|\InfMat\|_p\triangleq\sup_{v\neq 0}\frac{\|\InfMat {v}\|_p}{\|{v}\|_p}.
\]
We define the following condition for the operator norm of Dobrushin's influence matrix to be arbitrarily bounded, but not necessarily by a constant:
\begin{condition}
\label{cond:main}
There is a $p\in[1,\infty]$ such that 
$\|\InfMat\|_p\le \rho$ for some $\rho>0$.
\end{condition}

It is well known that contraction of Dobrushin's influence matrix in operator norms implies the mixing of the chain.  
Specifically, $\|\InfMat\|_1=\max_{v}\sum_{u}\Inf{u}{v}<1$ corresponds to the Dobrushin condition~\cite{dobrushin1970prescribing}, and $\|\InfMat\|_\infty=\max_{u}\sum_{v}\Inf{u}{v}<1$ corresponds to the Dobrushin-Shlosman condition~\cite{dobrushin1985completely,dobrushin1985constructive}, both of which imply optimal mixing times.
In a seminal work~\cite{hayes2006simple}, the $\ell_2$-Dobrushin's condition $\|\InfMat\|_2<1$ was studied,
which can be related to several key approaches for mixing, including coupling~\cite{hayes2006simple}, spectral gap~\cite{eldan2021spectral}, log-Sobolev~\cite{marton2019logarithmic} and modified log-Sobolev inequalities~\cite{blanca2022mixing}.
For spin systems, this condition was generalized by Dyer, Goldberg, and Jerrum to $\|\InfMat\|<1$ for arbitrary matrix norms~\cite{dyer2009matrix}.

\paragraph{Main Theorems.}
Let $(X_t)_{t\in\mathbb{R}_{\ge0}}$ denote a continuous-time single-site dynamics on the state space $Q^V$, 
defined by local update distributions ${P_v^{\tau}}$ specified on the graph $G = (V, E)$.
The single-site dynamics are presented to the algorithm as oracles for evaluating the distributions $P_v^{\tau}$. 
Specifically, given any $v \in V$, $\tau \in Q^{N_v^+}$, and $x \in Q$, the oracle returns the probability value $P_v^{\tau}(x)$. 
We present a parallel algorithm in \CRAM{} for simulating single-site dynamics.

\begin{theorem}[informal]\label{main-thm:parallel} 
For any single-site dynamics on a graph $G=(V,E)$ that satisfies \Cref{cond:main} with parameter~$\rho$,
there exists a parallel algorithm that can simulate the continuous-time chain $(X_t)_{0\le t\le T}$ from any initial state $X_0\in Q^V$ up to any time $0 < T < \mathrm{poly}(|V|)$,
with ${O}\left(\rho\cdot (T+\log|V|)\right)$ depth on $\tilde{O}\left(|E|+|V||Q|^2\right)$ processors with high probability.
\end{theorem}
The algorithm is a Las Vegas algorithm that guarantees to return a {faithful} copy of $X_T$ upon termination, while the complexity bounds hold with high probability ($1-|V|^{-O(1)}$). 
The $\tilde{O}(\cdot)$ notation hides poly-logarithmic factors.

Furthermore, when a stronger Dobrushin condition is satisfied, 
essentially the same parallel algorithm for perfect simulation can terminate much faster, 
achieving an exponential speedup. This is showcased by the following theorem.
\begin{theorem}[informal]\label{main-thm:two-spin} 
For any single-site dynamics on graph $G=(V,E)$ with domain size $|Q|=2$ satisfying \Cref{cond:main} with parameter $\rho<1$,
there exists a parallel algorithm that can simulate the continuous-time chain $(X_t)_{0\le t\le T}$ from any initial state $X_0\in Q^V$ up to any time $0 < T < \mathrm{poly}(|V|)$,
with $O(\frac{1}{1-\rho}(\log T+\log |V|))$ depth on $O\left(|E|\cdot T\right)$ processors with high probability.
\end{theorem}

\begin{remark}
We note that when applying the parallel simulation for the purpose of drawing approximate samples, the exponential speedup achieved in \Cref{main-thm:two-spin} may not offer a significant advantage over \Cref{main-thm:parallel}.
This is because, as show in~\cite{dyer2009matrix}, for spin systems, under the condition $\|\InfMat\|_p < 1$, the continuous chain mixes within time $T=O(\log n)$. 
With such $T$, the complexity bounds in \Cref{main-thm:parallel,main-thm:two-spin} are essentially the same.

Nevertheless, we still present \Cref{main-thm:two-spin} separately,
because it holds more generally beyond spin systems, showing that super-linear speedup is achievable for perfect simulation of single-site dynamics, which is conceptually significant.
\end{remark}

The parallel algorithm described above can also be implemented in the \CONGEST{} model, providing an efficient distributed algorithm for simulating single-site dynamics.

\begin{theorem}[informal]\label{main-thm:distributed}
For any single-site dynamics on a graph $G=(V,E)$ that satisfies \Cref{cond:main} with parameter~$\rho$,
there exists a distributed algorithm that can simulate the continuous-time chain $(X_t)_{0\le t\le T}$ from any initial state $X_0\in Q^V$ up to any time $0 < T < \mathrm{poly}(|V|)$,
within ${O}\left(\rho\cdot (T+\log |V|)\right)$ rounds on the network $G$, using messages each of $O(\log |V|\cdot \log |Q|)$ bits.
\end{theorem}
In the \CONGEST{} model, where each node has unlimited computational power and local memory, the update distributions $P_v^\tau$ can be provided to each node $v$ as local input.
The algorithm is a Monte Carlo algorithm that terminates in a fixed number of rounds and succeeds with high probability.

\paragraph{Useful Consequences.}
Note that in the above theorems, 
$T$ represents the continuous-time duration, 
which corresponds to $N$ discrete update steps where $N \sim \mathrm{Pois}(nT)$, with $n = |V|$. 
According to a generic lower bound~\cite{hayes2007general}, 
single-site dynamics requires $\Omega(n \log n)$ discrete steps to mix. 
Consequently, when applied to such mixing chains, the above algorithms achieve a parallel speedup of $\Theta(n)$.

For example, for any \emph{reversible} and \emph{ergodic} (irreducible and aperiodic) single-site dynamics, if the Dobrushin-Shlosman  condition  $\|\rho\|_\infty<1$ holds (which implies contraction of path coupling), 
the chain mixes within $n^{-O(1)}$ total variation distance to the stationary distribution $\pi$ in $O(n\log n)$ steps or $O(\log n)$ continuous time.
Combined with \Cref{main-thm:parallel},
this leads to the following straightforward corollary for MCMC sampling in the \RNC{} class (problems solvable by randomized parallel algorithms with poly-logarithmic depth and polynomial processors).

\begin{corollary}\label{cor:general}
For any reversible and ergodic single-site dynamics with \RNC{}-computable initial states and local update distributions, if the Dobrushin’s influence matrix satisfies $\|\InfMat\|_\infty < 1$, there exists an \RNC{} algorithm that can draw approximate samples (within any $n^{-O(1)}$ total variation distance) from the stationary distribution~$\pi$.
\end{corollary}

One can replace the above condition $\|\InfMat\|_\infty<1$ in \Cref{cor:general} with any sufficient condition that implies both $\tilde{O}(n)$ mixing time and \Cref{cond:main} with $\rho=O(1)$. In such cases, the corollary still holds.
The same applies to any sufficient condition that implies both $\mathrm{poly}(n)$ mixing time and \Cref{cond:main} with $\rho<1$.
However, in many cases, especially for spin systems, the mixing time bound may already be implied by \Cref{cond:main} with $\rho < 1$~\cite{dyer2009matrix}.

\begin{remark}
\Cref{cond:main} with $\rho=O(1)$ is significantly weaker than known Dobrushin conditions for mixing or known sufficient conditions for efficient parallelizing single-site dynamics. 
For example, \Cref{cond:main} with $\rho=O(1)$ holds if a weakened Dobrushin-Shlosman condition, $\|\rho\|_\infty=O(1)$, is satisfied.  
This condition essentially means that  the discrepancy in path coupling can grow by at most an $O(1)$ factor per step,
implying that ``disagreements do not propagate super-exponentially.''

For various graphical models, such as the Ising model, hardcore model, and proper coloring on graphs with bounded maximum degree, this condition is no stronger than the necessary conditions for the chains to mix.
Moreover, it is much weaker than the Lipschitz condition in \cite{feng2021distributed} for parallelizing Metropolis chains, 
which would imply that $\|\InfMat\|_1 = O(1)$ for the influence matrix $\InfMat$ of the Metropolis chain.
\end{remark}

\paragraph{Hardcore and Ising samplers.}
Graphical models can represent high-dimensional (Gibbs) distributions using local factors.
Let $G=(V,E)$ be an undirected graph. Let $\lambda>0$ be the fugacity (or external field, vertex activity) and $\beta>0$ be the temperature (or edge activity).
The hardcore Gibbs distribution, denoted $\mu^{\sf{hardcore}}$, is defined over all $\sigma\in\{0,1\}^V$ indicating independent sets in $G$ as follows:
\[
\forall \text{ independent set }\sigma\in\{0,1\}^V,\quad \mu^{\sf{hardcore}}(\sigma)\propto\lambda^{\|\sigma\|_1}.
\]
The Ising Gibbs distribution, denoted $\mu^{\sf{Ising}}$, is defined over all $\sigma\in\{0,1\}^V$ as follows:
\[
\forall \sigma\in\{0,1\}^V,\quad \mu^{\sf{Ising}}(\sigma)\propto\beta^{m(\sigma)}\lambda^{\|\sigma\|_1},
\]
where $m(\sigma)\triangleq\sum_{\{u,v\}\in E}I[\sigma_u=\sigma_v]$ counts the number of monochromatic edges in $\sigma$.

The normalizing factors of these Gibbs distributions are known as the \emph{partition functions}, which are essential for counting algorithms.
For graphs $G$ with maximum degree $\Delta=\Delta_G\ge 3$, 
the  \emph{uniqueness conditions} for the hardcore and Ising models are respectively given by $\lambda<\lambda_c(\Delta)\triangleq\frac{(\Delta-1)^{\Delta-1}}{(\Delta-2)^{\Delta}}$ and $\beta\in\left(\frac{\Delta-2}{\Delta},\frac{\Delta}{\Delta-2}\right)$.
Beyond these regimes, either the single-site dynamics is slow mixing or the sampling problem itself becomes computationally intractable~\cite{gerschcnfeld7reconstruction,sly2014counting,galanis2016inapproximability}.

Under the uniqueness condition,
optimal $O(n\log n)$ mixing times for the Glauber dynamics have been established for these models through a series of breakthroughs~\cite{chen2022optimal, chen2022localization, anari2022entropic, chen2021rapid, chen2021optimal, chen2020rapid, anari2020spectral}.
Additionally,  for the hardcore model with $\lambda=O\left(\frac{1}{\Delta}\right)$ and the Ising model with $\beta\in 1\pm O\left(\frac{1}{\Delta}\right)$,
\Cref{cond:main} holds with $\rho=O(1)$.
Consequently, we obtain the following corollary.
\begin{corollary}[Hardcore and Ising samplers]
\label{cor:Ising}
For graphs $G=(V,E)$ with maximum degree $\Delta=\Delta_G\ge 3$, 
there exist parallel samplers for the Gibbs distributions of the following models:
\begin{itemize}
    \item hardcore model with fugacity $\lambda\le(1-\delta)\lambda_c(\Delta)=(1-\delta)\frac{(\Delta-1)^{\Delta-1}}{(\Delta-2)^{\Delta}}$ for $\delta\in(0,1)$,
    \item Ising model with temperature $\beta\in\left(\frac{\Delta-2+\delta}{\Delta-\delta},\frac{\Delta-\delta}{\Delta-2+\delta}\right)$  for $\delta\in(0,1)$,
\end{itemize}
such that an approximate sample within $\frac{1}{\mathrm{poly}(|V|)}$ total variation distance from the Gibbs distribution can be obtained within  $O_\delta(\log\Delta\cdot \log |V|)$ depth on $\tilde{O}(|E|+|V|)$ processors.
\end{corollary}

To the best of our knowledge, these are the first \RNC{} samplers for any graphical model with unbounded maximum degree up to the critical condition.
Similar results can be obtained for general two-state anti-ferromagnetic spin systems under the critical conditions characterized in~\cite{chen2022optimal}.

The parallel samplers can be implemented as ${O}(\log n)$-round \CONGEST$(\log n)$ algorithms on the network $G$, where the local update distributions can be evaluated using 1-round communication.

The parallel samplers in \Cref{cor:Ising} can be transformed into parallel approximate counting algorithms using {non-adaptive} simulated annealing.
By applying the Chebyshev's cooling schedule from \cite{dyer1991random} and \cite[Corollary 2.4]{vstefankovivc2009adaptive}, the hardcore and Ising samplers in \Cref{cor:Ising} can be transformed into ${O}(\log n\cdot\log\Delta)$-depth algorithms on $\tilde{O}(mn^2/\epsilon^2)$ processors for $\epsilon$-approximations of the hardcore and Ising partition functions, under the same respective uniqueness conditions.
To the best of our knowledge, these are the first  \RNC{} approximation algorithms for the partition function of any graphical model with unbounded maximum degree.

\paragraph{SAT sampler.}
Consider conjunctive normal form (CNF) formulas over $n$ Boolean variables. 
A CNF formula  $\Phi$  is called a  \emph{$(k,d)$-formula} if all of its clauses have size $k$ and each variable appears in at most $d$ clauses.
In \cite{feng2021fast,moitra2019approximate}, a $\mathrm{poly}(d,k)\cdot n^{1+2^{-20}}$-time algorithm is presented for sampling almost uniform satisfying solutions, under a local lemma condition $k \geq 20\log k + 20\log d + 60$.
Here we show that this algorithm  can be parallelized within the \RNC{} class.

\begin{corollary}[SAT sampler]
\label{cor:CNF}
There is a parallel algorithm such that
for any $(k,d)$-formula $\Phi$ satisfying 
\begin{align}
k \geq 20\log k + 20\log d + 60,\label{eq:LLL-cond}
\end{align}
the algorithm outputs an almost uniform satisfying assignment for $\Phi$ within $\frac{1}{\mathrm{poly}(n)}$ total variation distance.
This is achieved using  $\mathrm{polylog}(n)$ depth on $\mathrm{poly}(d, n)$ processors, where $n$ is the number of variables in $\Phi$.
\end{corollary}
The Glauber dynamics is implemented on a  more complicated joint distribution $\mu_M$,  constructed by projecting the uniform distribution $\mu$ over all SAT solutions of $\Phi$  onto a carefully chosen subset of variables $M\subset V$. 
This projected Markov chain was introduced because the original solution space may be disconnected through local Markov chains.
Tighter analyses in the follow-up works \cite{FHY21,JPV21,HSW21} have proved the rapid mixing and efficient simulations of the projected Markov chains under weaker local lemma conditions, with improved constants, than the one in~\eqref{eq:LLL-cond}.
Similar improved bounds can also be achieved by the \RNC{} samplers using more involved analyses.
Similarly, by the non-adaptive annealing provided in \cite[Algorithm 7]{feng2021fast}, we obtain a $\mathrm{polylog}(n)$-depth  $\mathrm{poly}(d, n)$-processor algorithm for approximately counting the number of satisfying solutions of any $(k,d)$-CNF satisfying \eqref{eq:LLL-cond}.

Proofs of \Cref{cor:Ising} and \ref{cor:CNF} are straightforward and provided in \Cref{sec:applications} for completeness.

\subsection{Related Work in Correlated Sampling}

A key step in our algorithm for parallel simulation of single-site dynamics relies on generating each single-site update with a ``universal coupling'' of randomness (formally defined in~\Cref{def:universal-coupling}),
which synchronizes all local update distributions $P_v^\tau$ across different neighborhood configurations $\tau$.
While this coupling does not change the definition or sequential simulation of the Markov chain, it is essential for achieving efficient parallel simulation while maintaining the accuracy of the simulation. 
Notably, it overcomes the barrier identified in~\cite{teng1995independent}.

Perhaps due to its natural definition, 
this machinery of universally coupling distributions over the same sample space 
has actually been explored in various independent contexts to solve diverse problems.
These include MinHash sketching~\cite{broder1997resemblance,charikar2002similarity}, rounding linear programming relaxations~\cite{kleinberg2002approximation}, parallel repetition of 2-player 1-round games~\cite{holenstein2007parallel,rao2008parallel,barak2008rounding}, cryptography~\cite{rivest2016symmetric}, and replicability and differential privacy of learning~\cite{bun2023stability,ghazi2021user,kalavasis2023statistical,karbasi2023replicability}.

The problem is now commonly referred as \emph{correlated sampling} and was formalized in \cite{bavarian2020optimality} as follows.

\begin{definition}[correlated sampling]
A correlated sampling strategy for a finite sample space $\Omega$ with error rate $\epsilon:[0,1]\to [0,1]$ is a procedure $\Sample: \Delta(\Omega)\times [0,1] \to \Omega$, such that
\begin{itemize}
    \item \textbf{Correctness:} $\quad \forall p\in \Delta(\Omega), x\in \Omega$, $\Pr_{\RandSeed  \sim [0,1]}\left[\Sample(p,\RandSeed) = x\right] = p(x)$;
    \item \textbf{Error rate:} $\quad \forall p,q\in \Delta(\Omega) \text{ with } \dtv(p,q) = \delta$,  $\Pr_{\RandSeed  \sim [0,1]}\left[\Sample(p,\RandSeed) \neq \Sample(q,\RandSeed)\right] \leq \epsilon(\delta)$.
\end{itemize}
In above, $\Delta(\Omega)\triangleq \left\{p\in[0,1]^\Omega \bigm{|} \mbox{$\sum_{x\in\Omega}p(x)=1$}\right\}$ denotes the probability simplex on the sample space $\Omega$.
\end{definition}

In particular, the correlated sampling strategy used in our paper achieves an optimal error rate $\epsilon(\delta) = 2\delta/(1+\delta)$, 
and was independently discovered in \cite{kleinberg2002approximation} and \cite{holenstein2007parallel}.
This will be discussed in detail in \Cref{sec:universal-coupling}.

To the best of our knowledge, the present work is the first to apply the correlated sampling method to parallelize stochastic processes.
Very recently, Anari, Gao, and Rubinstein~\cite{anari2024parallel} generalized and applied the correlated sampling approach to derive a  parallel sampler for arbitrary high-dimensional distributions, achieving sub-linear depth and polynomial total work, assuming the availability of a counting oracle.


\section{Preliminaries}

\paragraph{Single-site dynamics.}
Let $V$ be a set of $n$ \emph{sites}, and let $Q=\{1,2,\ldots,q\}$ be a finite set of $q\ge 2$ \emph{spins}.
A \emph{configuration} on $S\subseteq V$ is an assignment $\sigma\in Q^S$ of spins to the sites in $S$.

Let $G=(V,E)$ be an undirected graph on vertex set $V$.
For each vertex $v\in V$, let $N_v\triangleq\{u\in V\mid \{u,v\}\in E\}$ denote the neighborhood of $v$ in $G$, and let $N_v^+\triangleq N_v\cup\{v\}$ denote the \emph{inclusive neighborhood} of $v$ in $G$.

We use the term \emph{single-site dynamics} on graph $G$ to refer to a specific type of Markov chain on space $Q^V$.
Let $\{P_v^\tau\}$ be a collection of \emph{local update distributions}, 
such that for every site $v\in V$ and every configuration $\tau\in N_v^+$ on $v$'s inclusive neighborhood, $P_v^\tau$ is a distribution over all spins in $Q$.
A Markov chain on space $Q^V$ is specified by $\{P_v^\tau\}$.
At the current configuration $\sigma\in Q^V$, the chain transitions as follows:
\begin{enumerate}
\item
Pick $v\in V$ uniformly at random.
\item
Replace the spin $\sigma(v)$ of $v$ with a spin drawn independently according to $P_v^{\tau}$, where $\tau=\sigma(N_v^+)$.
\end{enumerate}
The collection $\{P_v^\tau\}$ is also called the \emph{local transition rule} for the chain.

\paragraph{Examples.}
Given a distribution $\mu$ over $Q^V$, for any $v\in V$ and any possible $\tau\in Q^{S}$ where $S\subseteq V\setminus\{v\}$, 
let  $\mu_v^{\tau}$ denote the 
marginal distribution at $v$ induced by $\mu$, 
given the configuration  on $S$ being fixed as $\tau$. Formally, for each $x\in Q$,
\[
\mu_v^{\tau}(x)=\Pr_{\sigma\sim\mu}[\sigma_v=x\mid \sigma_S=\tau].
\]

Consider a \emph{Gibbs distribution} $\mu$ defined on graph $G$ over sample space $Q^V$,
such that for any $v\in V$ and any possible $\sigma\in Q^{V}$, the marginal distribution $\mu_v^{\sigma_{V\setminus\{v\}}}$ depends only on the values of $\sigma$ at $v$'s neighbors. That is, $\mu_v^{\sigma_{V\setminus\{v\}}}=\mu_v^{\sigma_{N_v}}$.
Examples include:
\begin{itemize}
\item \emph{Hardcore model:} Let $\lambda>0$. A Gibbs distribution $\mu$ is defined over all such $\sigma\in \{0,1\}^V$ such that $\mu(\sigma) \propto \lambda^{\|\sigma\|_1}$ if $\{v\in V: \sigma_v=1\}$ forms an independent set in $G$, and $\mu(\sigma)=0$ otherwise.
\item \emph{Ising/Potts/coloring models:} 
Let $\beta\ge 0$. A Gibbs distribution $\mu$ is defined over all $\sigma\in Q^V$, 
such that $\mu(\sigma) \propto \beta^{m(\sigma)}$,
where $m(\sigma)$ 
denotes the number of monochromatic edges in $\sigma$.
\end{itemize}
\noindent
Two basic classes of single-site dynamics are:
\begin{itemize}
\item \emph{Glauber dynamics}: (a.k.a.~\emph{Gibbs sampler}, \emph{heat-bath dynamics}):
The local update distributions are $P_v^\tau=\mu_v^{\tau_{N_v}}$.
\item \emph{Metropolis chain}: To update $v$'s spin in  $\sigma\in Q^V$, first propose to replace $\sigma(v)$ with a spin $x\in Q$ chosen uniformly at random, and then accept it with probability $\min\left\{1,{\mu_v^{\sigma_{N_v}}(x)}/{\mu_v^{\sigma_{N_v}}(\sigma_v)}\right\}$.
Formally, the local update distributions are defined as $P_v^\tau(x)=\frac{1}{|Q|}\min\left\{1,{\mu_v^{\tau_{N_v}}(x)}/{\mu_v^{\tau_{N_v}}(\tau_v)}\right\}$ for $x\neq \tau_v$ and $P_v^\tau(\tau_v)=1-\sum_{x\neq \tau_v}P_v^\tau(x)$.
\end{itemize}
It is well known that both chains are reversible with respect to the stationary distribution $\mu$~\cite{levin2017markov}.
These chains can also be defined for general (not necessarily Gibbs) distributions $\mu$ over $Q^V$ 
by setting $G$ as the complete graph, thus making the marginal distributions $\mu_v^{\tau_{N_v}}=\mu_v^{\tau_{V\setminus\{v\}}}$.

For Gibbs distributions $\mu$ defined by hard constraints, where infeasible configurations have zero probability,
one can routinely extend the definition of $\mu_v^{\tau}$ to infeasible conditions $\tau$ by ignoring the constraints violated locally by $\tau$. 
For instance, for the uniform distribution $\mu$ over proper $q$-colorings of $G$, if $\tau\in Q^{N_v}$ is already improper, $\mu_v^{\tau}$ can still be defined as the uniform distribution over all available colors in $Q\setminus\{\tau_u\mid u\in N_v\}$, provided $q>\Delta_G$.
This ensures that the single-site dynamics may absorb to feasible states even if starting from an infeasible one.

Therefore, in this paper, we assume that for a single-site dynamics the local update distributions $\{P_v^\tau\}$ are defined for all sites $v\in V$ and neighborhood configurations $\tau\in Q^{N_v^+}$.

\paragraph{The continuous-time chain.}
The continuous-time variant of the single-site dynamics is defined as follows:
\begin{enumerate}
\item
Time progresses continuously from 0,
with each site $v\in V$ associated with an independent rate-1 Poisson clock.
Let $S_i=\sum_{1\leq j\leq i}T_j$, where $T_1,T_2,\dots$ are independent and identically distributed exponential random variables of rate 1. 
The Poisson clock rings at time $S_1,S_2,\ldots$.
\item
Whenever the clock at a site $v\in V$ rings, the spin of $v$ is updated according to the local transition rule  $\{P_v^\tau\}$ as in the discrete-time chain.
\end{enumerate}

Let $(X_t^{\mathsf{C}})_{t\in\mathbb{R}_{\ge 0}}$ denote this continuous-time process, 
and let $(X_t^{\mathsf{D}})_{t\in\mathbb{N}_{\ge 0}}$ denote its discrete-time counterpart.
The two processes are equivalent up to a speedup factor  $n=|V|$. 
\begin{proposition}
\label{pro:poisson1}
Conditioning on the same initial configuration $X_0^{\mathsf{C}}=X_0^{\mathsf{D}}\in Q^V$, 
for any $T> 0$, 
$X_T^{\mathsf{C}}$ is identically distributed as $X_N^{\mathsf{D}}$, 
where $N$ follows the Poisson distribution $\mathrm{Pois}(nT)$, where $n=|V|$.
\end{proposition}

\paragraph{Mixing time.}
According to the Markov chain convergence theorem~\cite{levin2017markov}, 
an irreducible and aperiodic Markov chain on a finite state space converges to a unique stationary distribution $\pi$.
The {mixing time} 
measures the rate at which the chain converges to this stationary distribution.
For two distributions $\mu,\nu$ over the same sample space $\Omega$, the \emph{total variation distance} between $\mu$ and $\nu$ is defined as $\dtv(\mu,\nu) \triangleq \frac{1}{2}\sum_{x\in \Omega}|\mu(x)-\nu(x)|$.
For a chain $X_t$ on the space $Q^V$ that converges to the stationary distribution $\pi$,
and for any configuration $\sigma\in Q^V$, let $p_{t}^{\sigma}$ denote the distribution of $X_t$, given $X_0=\sigma$.
The \emph{mixing time} is defined as $t_{\mathsf{mix}}(\epsilon)\triangleq \max_{\sigma\in Q^V}\min\left\{t>0\mid \dtv\left(p_t^\sigma,\pi\right)\le\epsilon\right\}$.

Let $t_{\mathrm{mix}}^{\mathsf{C}} (\epsilon)$ and $t_{\mathrm{mix}}^{\mathsf{D}} (\epsilon)$  denote the mixing times of the continuous-time single-site dynamics $(X_t^{\mathsf{C}})_{t\in\mathbb{R}_{\ge 0}}$  and its discrete-time counterpart  $(X_t^{\mathsf{D}})_{t\in\mathbb{N}_{\ge 0}}$, respectively, where both are defined by the same local transition rule $\{P_v^{\tau}\}$. 
The following relation between the two mixing times follows immediately from \Cref{pro:poisson1}:
\begin{align}
n\cdot t_{\mathrm{mix}}^{\mathsf{C}} (\epsilon) = O\left(t_{\mathrm{mix}}^{\mathsf{D}}(\epsilon/2) + \log \frac{1}{\epsilon}\right),\label{eq:continuous-discrete-mixing}
\end{align}
where $n=|V|$ is the number of sites.

In particular, an $O(n\log n)$ mixing time for the discrete-time chain implies an  $O(\log n)$ mixing time for the continuous-time chain.
Furthermore, due to a general lower bound \cite{hayes2007general}, these mixing times are optimal for single-site dynamics. 

\paragraph{Computation models.} 
In this paper, we assume the concurrent-read concurrent-write ($\mathsf{CRCW}$) parallel random access machine  (\PRAM) model with arbitrary write, 
where an arbitrary value written concurrently is stored.
The \emph{depth} of an algorithm in this model refers to the number of time steps required for its execution.

We use \NC{} to denote both the class of parallel algorithms with poly-logarithmic depth and polynomial processors and the class of problems solvable by such algorithms.
The class \RNC{} refers to the randomized counterpart of \NC{}. 

The \CONGEST{} model is defined on an undirected network $G=(V,E)$, where the nodes represent processors.
Initially, each node receives its local input and private random bits.
Communications are synchronized and proceed in rounds. 
In each round, each node may perform arbitrary local computation and send a $B$-bit message to each of its neighbors.
The messages are received by the end of the round. The model is denoted as \CONGEST(B).


\section{A Locally-Iterative Algorithm for Simulating Markov Chain}

We introduce a locally-iterative algorithm that simulates single-site dynamics.
The algorithm falls into the general framework of message passing algorithms, such as the well-known \emph{belief propagation (BP)} algorithms~\cite{pearl1982reverend,mezard2009information}.
In this algorithm, each site maintains an internal state, 
and all internal states are updated iteratively based on the current internal states within their local neighborhoods until a fixpoint is reached.

Let $(X_t)_{t\in\mathbb{R}_{\ge 0}}$ be a continuous-time single-site dynamics on a graph $G=(V,E)$, defined by local update distributions $\{P_v^\tau\}$.
Our goal is to simulate this Markov chain up to a fixed time $T>0$.
For each site $v\in V$, a rate-1 Poisson clock runs independently up to time $T$, generating a sequence of times:
\begin{align*}
0=t_0^v<t_1^v<\cdots <t_{m_v}^v<T, \text{where }m_v\sim\mathrm{Pois}(T).
\end{align*}
%
With probability 1, all $t_i^v$ values for $i\ge 1$ are distinct.
We refer to such a collection of time sequences for all sites  $\mathfrak{T}=(t_i^v)_{v\in V, 0\le i\le m_v}$ as  an \emph{update schedule} up to time $T$.

Given an update schedule $\mathfrak{T}$, we identify the $i$-th {update} at site $v$ (which occurs at time $t=t_i^v$) by the pair $(v,i)$.
To perform this update,  the current spin $X_t(v)$ at site $v$ is replaced by a random spin generated independently according to $P_v^{\tau}$. 
Here, $\tau\in Q^{N_v^+}$ is constructed as follows: for each $u\in N_v^+$, $\tau_u=X_{t_{j_u}^u}(u)$ represents the spin generated during the $j_u$-th update at site $u$, where $j_u=\max\{j\ge 0\mid t_j^u<t_i^v\}$.

The evolution of this process can then be described by the following abstract dynamical system:
\begin{align}
X_t(v)\gets \Sample\left(P_v^{\tau},\RandSeed_{(v,i)}\right),\label{eq:sample-subroutine-for-update}
\end{align}
where $\RandSeed_{(v,i)}$ denotes the \emph{random seed} used for the update $(v,i)$.
The function $\Sample\left(\mu,\RandSeed\right)$ is a \emph{deterministic} subroutine that,
given the description of any distribution $\mu$ over $Q$ and access to the random seed $\RandSeed$,
guarantees to return a spin distributed according to $\mu$ by utilizing the random bits in $\RandSeed$. 

Given the initial configuration $X_0\in Q^V$,
once the update schedule $\mathfrak{T}=(t_i^v)_{v\in V,0\le i\le m_v}$ up to time $T$
and the random bits $\mathfrak{R}=(\RandSeed_{(v,i)})_{v\in V,1\le i\le m_v}$ used in all updates are generated,
the entire evolution of the chain $(X_t)_{0\le t\le T}$ is uniquely determined by sequentially executing~\eqref{eq:sample-subroutine-for-update}.

A parallel procedure that  faithfully simulates this process is given in \Cref{alg:BP-simulation}.

\begin{algorithm}[ht]
 \SetKwInOut{Input}{Input} 
\Input{ initial configuration $X_0\in Q^V$; update schedule $\mathfrak{T}=(t_i^v)_{v\in V,0\le i\le m_v}$;
\\assignment $\mathfrak{R}=(\RandSeed_{(v,i)})_{v\in V,1\le i\le m_v}$ of random bits for resolving updates.}
{initialize} $\ell\gets 0$ and $\CurConfig{\widehat{X}}{0}{v}[i]\gets X_0(v)$ for all $v\in V$, $0\le i\le m_v$\;
\Repeat{$\widehat{X}^{(\ell)}=\widehat{X}^{(\ell-1)}$}{\label{alg:line:repeat-loop1}
	$\ell\gets \ell+1$\;
	\lForAllPar{$v\in V$}{$\CurConfig{\widehat{X}}{\ell}{v}[0]\gets X_0(v)$}
	\ForAllPar{updates $(v,i)$, where $v\in V$, $1\le i\le m_v$, }{
		let $\tau\in Q^{N_v^+}$ be constructed as: \hspace{300pt}
		\mbox{\hspace{20pt}}$\forall u\in N_v^+$, $\tau_u\gets\CurConfig{\widehat{X}}{\ell-1}{u}[j_u]$ for $j_u=\max\{j\ge 0\mid t_j^u<t_i^v\}$\; \label{alg:line:neighborhood-config}
		$\CurConfig{\widehat{X}}{\ell}{v}[i]\gets \Sample\left(P_v^{\tau},\RandSeed_{(v,i)}\right)$\; \label{alg:line:BP-update}
	}
} 
\caption{Locally-iterative algorithm for simulating single-site dynamics}\label{alg:BP-simulation}
\end{algorithm}

\begin{remark}[{consistent iterative updating}]\label{rmk:consistent-sampling}
The algorithm maintains an array $\widehat{X}^{(\ell)}$, where each entry $\widehat{X}_{v}^{(\ell)}[i]$ corresponds to the update $(v,i)$ in the chain.
In every iteration $\ell\ge 1$, each entry $\widehat{X}_{v}^{(\ell)}[i]$ is updated from $\widehat{X}^{(\ell-1)}$ according to the rule~\eqref{eq:sample-subroutine-for-update},
using the same assignment  of random bits $\RandSeed_{(v,i)}$ regardless of the iteration  $\ell$.
\end{remark}

\Cref{alg:BP-simulation} is a locally iterative parallel algorithm that simulates the abstract dynamical system described in~\eqref{eq:sample-subroutine-for-update}, utilizing coupled randomness as noted in \Cref{rmk:consistent-sampling} to accelerate convergence while ensuring accurate simulation.
Unlike the belief propagation (BP) algorithm used to calculate marginal probabilities, whose correctness is not guaranteed on general graphs~\cite{mezard2009information},  
\Cref{alg:BP-simulation} always converges to the correct chain within a finite number of steps.
This guarantee is based on the following two observations:
\begin{enumerate}
    \item \label{item:monotone-dynamical-system} The abstract dynamical system in~\eqref{eq:sample-subroutine-for-update} is acyclic due to its monotonicity over time. 
    \item  The correct evolution of the Markov chain $(X_t)_{0\le t\le T}$ from time 0 to time $T$ corresponds to a fixed point solution satisfying~\eqref{eq:sample-subroutine-for-update}. 
\end{enumerate}

We now elaborate on these observations.
Fix an update schedule $\mathfrak{T}=(t_i^v)_{v\in V,0\le i\le m_v}$.
For any updates $(u,j)$ and $(v,i)$ where $i,j\ge 1$, define the relation: 
\begin{equation}
\label{order-def}
(u,j)\prec (v,i) \iff u\in N_v^+ \land j=\max\{j\ge 0\mid t_j^u<t_i^v\}.
\end{equation}
Intuitively, $(u,j)\prec (v,i)$ means that  the update $(v,i)$ is determined  based on the outcome of update~$(u,j)$.
This relation $\prec$ between updates  naturally defines a directed graph $D_{\prec}=(U,E_{\prec})$, where the vertices are $U=\left\{ (v,i)\mid v \in V, 1\leq i\leq m_v \right\}$ and the edges are:
\begin{equation}
\label{dag-def}
    E_\prec=\left\{ \langle(u,j), (v,i)\rangle \mid (u,j),(v,i)\in U, (u,j)\prec (v,i) \right\}.
\end{equation}
It is easy to verify that $D_{\prec}$ is a directed acyclic graph (DAG) due to the monotonicity in time $t_i^v$. We refer to $D_{\prec}$ as the \emph{dependency graph}, which depicts the dependencies between updates.

The following lemma is established through structural induction on the dependency graph $D_{\prec}$.
\begin{lemma}[{convergence \& correctness}]\label{lemma:BP-convergence}
Fix any initial configuration $X_0\in Q^V$, update schedule $\mathfrak{T}=(t_i^v)_{v\in V,0\le i\le m_v}$ up to time $T$, and assignment of random bits  $\mathfrak{R}=(\RandSeed_{(v,i)})_{v\in V,1\le i\le m_v}$ for resolving updates.
Let $(X_t)_{0\le t\le T}$ be the continuous-time chain fully determined by $X_0$, $\mathfrak{T}$ and $\mathfrak{R}$. 
\begin{enumerate}
\item \label{item:BP-convergence}
\Cref{alg:BP-simulation} terminates within $m+1$ iterations (of the \textbf{repeat} loop), where $m=\sum_{v\in V}m_v$.
\item \label{item:BP-correctness}
Upon termination, 
$\widehat{X}$ correctly gives all transitions of the chain $(X_t)_{0\le t\le T}$,
such that $\widehat{X}_v[i]=X_{t_i^v}(v)$ for every update $(v,i)$. 
In particular, $\left(\widehat{X}_v[m_v]\right)_{v\in V}=X_T$ gives the configuration at time $T$.
\end{enumerate}
\end{lemma}

\paragraph{Fast convergence}
For a random pair $(\mathfrak{T},\mathfrak{R})$ generated as in a continuous-time chain $(X_t)_{0\le t\le T}$ up to time $T$,
calculations from~\cite{feng2021distributed,hayes2007general} show that,  with high probability, 
the length of any path in the dependency graph $D_{\prec}$ is bounded by $O(\Delta T+\log n)$, 
where $\Delta$ is the maximum degree of the underlying graph $G$.
Consequently, \Cref{alg:BP-simulation} returns within $O(\Delta T+\log n)$ iterations with high probability.

Indeed, 
\Cref{alg:BP-simulation}  can converge to the correct fixpoint significantly faster than the length of the longest path in $D_{\prec}$.
To understand this, recall that in \Cref{alg:line:BP-update} of \Cref{alg:BP-simulation}, the array $\widehat{X}$ is updated as follows:
\[
{\widehat{X}}_{v}[i]\gets \Sample\left(P_v^{\tau},\RandSeed_{(v,i)}\right).
\]
Here, the entry $\widehat{X}_{v}[i]$ (which corresponds to the update $(v,i)$) can locally stabilize (i.e.~remains the same in two consecutive iterations) even before its neighborhood configuration $\tau$ has stabilized, under the following conditions:
\begin{enumerate}
\item\label{informal-Lipschitz-cond-1}
The distributions $P_v^{\tau}, P_v^{\tau'}$ are close to each other when their boundary conditions $\tau, \tau'\in Q^{N_v^+}$ are similar.
\item\label{informal-Lipschitz-cond-2}
For  $P_v^{\tau}, P_v^{\tau'}$ that are close to each other,
it is likely that $\Sample(P_v^{\tau},\RandSeed_{(v,i)})=\Sample(P_v^{\tau'},\RandSeed_{(v,i)})$.
\end{enumerate}
The first property is captured by \Cref{cond:main},
and the second is formalized by the following notion of ${\alpha}$-competitiveness:

\begin{definition}[{${\alpha}$-competitiveness}]\label{def:competitive-sample}
Let $\alpha\ge 1$.
A procedure $\Sample(\mu,\RandSeed)$ that returns a sample from distribution $\mu$ arbitrarily specified over $Q$, 
is said to be \emph{$\alpha$-competitive}, if for any $\mu, \nu$ over $Q$,
\begin{align}
\Pr_{\RandSeed}\left[\,\Sample\left(\mu,\RandSeed\right)\neq \Sample\left(\nu,\RandSeed\right)\,\right]\le \alpha\cdot \dtv(\mu,\nu).\label{eq:competitive-sample}
\end{align}
\end{definition}

The following convergence bound for \Cref{alg:BP-simulation} on random input justifies the above intuition.
On random choices $(\mathfrak{T},\mathfrak{R})$ generated as in the continuous-time chain $(X_t)_{0\le t\le T}$ up to time $T$, 
\Cref{alg:BP-simulation} faithfully simulates the chain using $O\left(T+\log {n}\right)$ iterations with high probability,
assuming \Cref{cond:main} with parameter $\rho=O(1)$ and an $O(1)$-competitive $\Sample$ subroutine.

\begin{lemma}[{fast convergence with linear speedup}]\label{lemma:BP-fast}
If \Cref{cond:main} holds with parameter $\rho$ and the $\Sample$ subroutine is $\alpha$-competitive,
then for every $T>0$ and $\epsilon\in(0,1)$,
\Cref{alg:BP-simulation} on any initial configuration $X_0\in Q^V$ and  random $(\mathfrak{T},\mathfrak{R})$ up to time $T$ terminates within $O\left(\alpha\rho\cdot T+\log \left(\frac{n}{\epsilon}\right)\right)$ iterations with probability $\ge 1-\epsilon$.
\end{lemma}

Furthermore, the following lemma gives a much improved bound on the convergence rate (with an exponential speedup) for \Cref{alg:BP-simulation}  under a more restricted condition $\alpha\rho<1$.

\begin{lemma}[{fast convergence with exponential speedup}]\label{lemma:BP-fast2}
Assume the condition of \Cref{lemma:BP-fast}.
If further $\alpha \rho<1$,
then for every $T>0$ and $\epsilon\in(0,1)$,
\Cref{alg:BP-simulation} on any initial configuration $X_0\in Q^V$ and  random $(\mathfrak{T},\mathfrak{R})$ up to time $T$ terminates within $O\left(\frac{1}{1-\alpha\rho}\log\left(\frac{nT}{\epsilon}\right)\right)$ iterations with probability $\ge 1-\epsilon$.
\end{lemma}

\paragraph{Universal coupling.} 
It remains to investigate whether there  exists a $\Sample$  procedure that always achieves small competitiveness when sampling from an arbitrarily specified distribution $\mu$ over $Q$.
Such a sampling procedure would provide a \emph{universal coupling},
simultaneously coupling all distributions $\mu$ over the same sample space $Q$.

We first show that for Boolean domain $Q$ of size $|Q|=2$, 
the standard \emph{inverse transform sampling} provides a $1$-competitive $\Sample$ subroutine that can perfectly couple all pairs of distributions $\mu$ over $Q$.
\begin{definition}[{inverse transform sampling}]\label{def:inverse-transform-sampling}
Let 
 $\RandSeed$ be uniformly distributed over $[0,1]$. For any distribution $\mu$ over the Boolean domain $Q =\{0,1\}$, define 
 \begin{align}
     \Sample(\mu,\RandSeed) = \begin{cases}0 & \RandSeed<\mu(0),\\ 1& \RandSeed \geq \mu(0). \end{cases}
 \label{eq:two-spin-procedure}
 \end{align}
\end{definition}
This $\Sample(\mu,\RandSeed)$ implements the {inverse transform sampling} method for the Bernoulli distribution $\mu$, 
where the uniform random seed $\RandSeed\in [0,1]$ is used as the quantile.

It is straightforward to verify that $\Sample(\mu,\RandSeed)$ is distributed as $\mu$, and 
for any distributions $\mu,\nu$ over $Q = \{0,1\}$, 
 \begin{align}
 \Pr_{\RandSeed}\left[\,\Sample\left(\mu,\RandSeed\right)\neq \Sample\left(\nu,\RandSeed\right)\,\right] = \dtv(\mu,\nu).
 \nonumber
 \end{align}
This confirms that the  $\Sample(\mu,\RandSeed)$ procedure as defined in \eqref{eq:two-spin-procedure} for sampling from Boolean domain  is $1$-competitive. 
As a result, we can state the following corollary of \Cref{lemma:BP-fast2} for Boolean domains. 

\begin{corollary}\label{lemma:BP-fast-2spin}
If $|Q|=2$ and \Cref{cond:main} holds with parameter $\rho<1$,
then for every $T>0$ and $\epsilon\in(0,1)$,
\Cref{alg:BP-simulation}, when applied to any initial configuration  $X_0\in Q^V$ and  random $(\mathfrak{T},\mathfrak{R})$ up to time $T$, 
terminates within $O\left(\frac{1}{1-\rho}\log\left(\frac{nT}{\epsilon}\right)\right)$ iterations with probability at least $1-\epsilon$.
\end{corollary}

For general finite non-Boolean domains $Q$, the competitive ratio achieved by the inverse transform sampling method increases linearly with $|Q|$.
However, it is surprising that a universal coupling within twice the optimal coupling can always be achieved between any pair of distributions $\mu,\nu$ over any finite sample space $Q$.

\begin{theorem}[{twice-optimal universal coupling}]
\label{thm:2-competitive}
For any finite sample space $Q$, there exists a $2$-competitive $\Sample(\mu,\RandSeed)$ procedure. 
\end{theorem}

This sampling procedure, now known as \emph{correlated sampling}, has been discovered in multiple independent works.
Notably, it was identified by Kleinberg and Tardos~\cite{kleinberg2002approximation} for rounding linear programs 
and by Holenstein~\cite{holenstein2007parallel} for parallel repetition of 2-player 1-round games.
The procedure will be formally described and analyzed in \Cref{sec:universal-coupling}
\color{black}

The rest of this section is dedicated to proving \Cref{lemma:BP-convergence}, \Cref{lemma:BP-fast} and \Cref{lemma:BP-fast2}.

\subsection{Worst-case Convergence and Correctness}\label{sec:BP-convergence}

The dependency graph $D_{\prec}$ defined in~\eqref{dag-def} gives rise to a notion of depth for updates.
For each update $(v,i)\in U$, its \emph{depth}, denoted by $\dep(v,i)$, is defined as follows:
\[
\dep(v,i)=
\begin{cases}
1 & \text{if }\neg\exists (u,j)\text{ s.t. }(u,j){\prec} (v,i),\\
1+\max\limits_{(u,j)\prec(v,i)}\dep(u,j) & \text{otherwise}.
\end{cases}	
\] 
Note that $\dep(v,i)$ represents the length of the longest path that ends at $(v,i)$ in the DAG $D_\prec$.

Now consider the $\widehat{X}$ maintained by \Cref{alg:BP-simulation}.
Note that for every update $(v,i)$ and any possible iteration number $\ell\ge 0$, the entry $\CurConfig{\widehat{X}}{\ell}{v}[i]$ is assigned only once in \Cref{alg:line:BP-update} of \Cref{alg:BP-simulation}.
For any $\ell\ge 1$ where the algorithm terminates before reaching the $\ell$-th iteration, we assume $\widehat{X}^{(\ell)}= \widehat{X}^{(\ell-1)}$.
Therefore, $\widehat{X}^{(\ell)}$ is well-defined for all $\ell\ge 0$, given the input $X_0$, $\mathfrak{T}$ and $\mathfrak{R}$.

The next lemma shows that for every update $(v,i)$, the corresponding entry $\CurConfig{\widehat{X}}{\ell}{v}[i]$ stops changing, i.e.~it locally terminates, after $\dep(v,i)$ iterations.
\begin{lemma}\label{lemma:round-depth-bound}
For each update $(v,i)$, where $v\in V$ and $1\leq i\leq m_v$, 
for any $\ell> \dep(v,i)$, it holds that
\[
\CurConfig{\widehat{X}}{\ell}{v}[i]=\CurConfig{\widehat{X}}{\ell-1}{v}[i].
\]
\end{lemma}

This immediately proves \Cref{item:BP-convergence} of \Cref{lemma:BP-convergence}, establishing the convergence of \Cref{alg:BP-simulation} within $m+1$ steps.
Since $m=\sum_{v\in V}m_v=|U|$ provides a trivial upper bound on the length of any path in $D_\prec=(U,E_\prec)$, \Cref{lemma:round-depth-bound} implies that the entire $\widehat{X}$ must have stopped changing after at most $m+1$ steps.

\begin{proof}[Proof of \Cref{lemma:round-depth-bound}]
A key observation is that in the $r$-th iteration, 
the neighborhood configuration $\tau\in Q^{N_v^+}$ constructed in \Cref{alg:line:neighborhood-config} for resolving an update $(v,i)$ satisfies the following: for every $u\in N_v^+$, $\tau_u= \CurConfig{\widehat{X}}{\ell-1}{u}[j_u]$, where either $(u,j_u)\prec (v,i)$ (and hence $\dep(u,j_u)\le\dep(v,i)$) or $j_u=0$. 
Intuitively, this observation indicates that every spin in the neighborhood configuration $\tau$ for resolving an update $(v,i)$ either results from a previous update or is part of the initial configuration.

We now prove the lemma by induction on depth.

Induction Basis:
For the updates $(v,i)$ with depth $\dep(v,i)=1$, there is no $(u,j)$ such that $(u,j)\prec(v,i)$.
Therefore, by the above observation,
the neighborhood configuration $\tau\in Q^{N_v^+}$ constructed in \Cref{alg:line:neighborhood-config} for resolving the update $(v,i)$ always satisfies $\tau_u= \CurConfig{\widehat{X}}{\ell-1}{u}[0] =X_0(u)$ for every $u\in N_v^+$.
This implies that $\CurConfig{\widehat{X}}{\ell}{v}[i] \gets \Sample\left(P_v^{\tau},\RandSeed_{(v,i)}\right)$ stops changing after the first iteration, proving the lemma for $\dep(v,i)=1$.

Induction Step:
Now, assume the lemma holds for all updates $(u,j)$ with $\dep(u,j)<\dep(v,i)$,
and consider an update $(v,i)$ with $\dep(v,i)>1$.
By the same observation and the induction hypothesis, 
the neighborhood configuration $\tau\in Q^{N_v^+}$ constructed in \Cref{alg:line:neighborhood-config} for resolving the update $(v,i)$ stops changing after $\dep(v,i)-1$ iterations.
Consequently, $\CurConfig{\widehat{X}}{\ell}{v}[i] \gets \Sample\left(P_v^{\tau},\RandSeed_{(v,i)}\right)$ stops changing after $\dep(v,i)$ iterations, finishing the induction.
\end{proof}

\label{sec:BP-correctness}
We now prove \Cref{item:BP-correctness} of \Cref{lemma:BP-convergence}, which addresses the correctness of \Cref{alg:BP-simulation}.
Let $\widehat{X}$ denote the $\widehat{X}^{(\ell)}$ in \Cref{alg:BP-simulation} when the algorithm terminates.
Recall that $(X_t)_{0\le t\le T}$ represents the continuous-time chain, which is fully determined by $X_0$, $\mathfrak{T}$ and $\mathfrak{R}$.
\Cref{item:BP-correctness} of \Cref{lemma:BP-convergence} can be restated as follows:

\begin{lemma}\label{lemma:faithful-simulation}
For each update $(v,i)$, where $v\in V$ and $1\leq i\leq m_v$, it holds that
\[
\widehat{X}_{v}[i]=X_{t_i^v}(v).
\]	
\end{lemma}
\begin{proof}	
As the fixpoint of \Cref{alg:BP-simulation}, 
the array $\widehat{X}$ satisfies that for every update $(v,i)$ where $v\in V$ and $1\leq i\leq m_v$,
\[
\widehat{X}_{v}[i] = \Sample\left(P_v^{\tau},\RandSeed_{(v,i)}\right),
\]
where $\tau\in Q^{N_v^+}$ is such that for every $u\in N_v^+$, $\tau_u=\widehat{X}_{u}[j_u]$ for $j_u=\max\{{j}\ge 0\mid t_j^u<t_i^v\}$.

By definition of the continuous-time chain $(X_t)_{0\le t\le T}$, for every update $(v,i)$ where $v\in V$ and $1\leq i\leq m_v$,
\[
X_{t_i^v}(v) = \Sample\left(P_v^{\tau},\RandSeed_{(v,i)}\right),
\]
where $\tau\in Q^{N_v^+}$ is such that for every $u\in N_v^+$, $\tau_u=X_{t_{j_u}^u}(u)$ for $j_u=\max\{{j}\ge 0\mid t_j^u<t_i^v\}$.

Additionally, $\widehat{X}_{v}[0]=X_0(v)$ for all $v\in V$. The lemma follows by inductively verifying the equality $\widehat{X}_{v}[i]=X_{t_i^v}(v)$ through the topological order of $D_\prec$ defined in \eqref{dag-def}.
\end{proof}

\subsection{Fast Average-case Convergence}\label{sec:BP-fast}
In a continuous-time chain $(X_t)_{t\in\mathbb{R}_{\ge 0}}$, the update schedule $\mathfrak{T}$ is randomly generated by independent Poisson clocks.
\Cref{alg:BP-simulation} terminates in significantly fewer steps with such a randomly generated update schedule.
\begin{lemma}[\cite{feng2021distributed,hayes2007general}]\label{lemma:unconditional-convergence-bound}
For a random update schedule $\mathfrak{T}$ where $\mathfrak{T}=(T_i^v)_{v\in V,0\le i\le M_v}$ is generated up to time $T$ by the $n$ independent rate-1 Poisson clocks on graph $G$ of maximum degree $\Delta$,
the probability that there exists an update $(v,i)$ with $\dep(v,i)\ge \ell$ is at most $n\cdot\left(\frac{\mathrm{e}(\Delta+1)T}{\ell}\right)^{\ell}$ for any $\ell\ge 1$.
\end{lemma}
\begin{proof}
Fix any path $(v_1,\ldots,v_{\ell})\in V^\ell$ where $v_{i+1}\in N_{v_i}^+$ for $1\le i<\ell$.
Consider the event where there exist $0<t_1<\cdots<t_\ell<T$ such that the Poisson clock at $v_i$ rings at time $t_i$.
As shown in \cite[Observation 3.2]{hayes2007general}, the probability of this event is at most  $\le\left(\frac{\mathrm{e}T}{\ell}\right)^{\ell}$.
By applying the union bound over all possible paths, where there are at most $\le n(\Delta+1)^{\ell}$ such paths, the probability that there is a path with this property is bounded by $n\cdot\left(\frac{\mathrm{e}(\Delta+1)T}{\ell}\right)^{\ell}$.
\end{proof}

Combining \Cref{lemma:round-depth-bound,lemma:unconditional-convergence-bound}, we conclude that for a randomly generated update schedule  $\mathfrak{T}$ up to time $T$, \Cref{alg:BP-simulation} will terminate within $\left\lceil2\mathrm{e}(\Delta+1)T+\log_2\left(\frac{n}{\epsilon}\right)\right\rceil$ iterations with probability at least $1-\epsilon$.

We now prove \Cref{lemma:BP-fast} and \Cref{lemma:BP-fast2}, establishing the faster convergence of \Cref{alg:BP-simulation} when the Dobrushin's influence matrix has a bounded operator norm.
We will actually prove slightly stronger results, for which we first need to introduce the following variant of Dobrushin's criterion.
\begin{definition}[Dobrushin's influence matrix for sampling]\label{def:sampling-matrix}
Given a single-site dynamics on a graph $G=(V,E)$ specified by $\{P_v^{\tau}\}$ and realized by a $\Sample(\cdot,\cdot)$ subroutine,
the \emph{Dobrushin's influence matrix for sampling} is a matrix $\SampleMat\in\mathbb{R}_{\ge 0}^{V\times V}$ defined by:
\begin{align}
\forall u,v\in V:\quad  \SampleMat(u,v)\triangleq \max_{\tau,\tau':\tau\oplus\tau'\subseteq\{u\}} \Pr\left[
\Sample(P_v^{\tau},\RandSeed) \neq \Sample(P_v^{\tau'},\RandSeed)
\right],\label{eq:sampling-matrix}
\end{align}
where the maximum is taken over all pairs of configurations $\tau,\tau'\in Q^{N_v^+}$ on the inclusive neighborhood $N_v^+=N_v\cup\{v\}$ of $v$,  such that $\tau$ and $\tau'$ agree with each other everywhere except at $u$.
\end{definition}

We also define the following variant of \Cref{cond:main}, tailored for the Dobrushin's influence matrix for sampling.
\begin{condition}
\label{cond:sampling}
There is a $p\in[1,\infty]$ such that 
$\|\SampleMat\|_p\le \theta$ for some $\theta>0$.
\end{condition}

It is straightforward to observe that $\SampleMat\le\alpha\InfMat$ entry-wise, assuming that the $\Sample(\cdot,\cdot)$ subroutine is $\alpha$-competitive (as defined in \Cref{def:competitive-sample}), where $\InfMat$ denotes the Dobrushin's influence matrix  (as defined in \Cref{def:dobrushin-matrix}).
Consequently, we have  $\|\SampleMat\|_p\le\alpha\|\InfMat\|_p$ for any $p\in[1,\infty]$, since $\SampleMat$ and $\InfMat$ are non-negative matrices and $\SampleMat\le\alpha\InfMat$ entry-wise.

Then, \Cref{lemma:BP-fast} and \Cref{lemma:BP-fast2} are implied by the following lemma.
\begin{lemma}[{fast convergence of \Cref{alg:BP-simulation}}]\label{lemma:BP-fast-sampling}
The following results hold for \Cref{alg:BP-simulation} on any initial configuration $X_0\in Q^V$ and and random $(\mathfrak{T},\mathfrak{R})$ up to time $T$:
\begin{enumerate}
    \item \label{lemma:BP-fast-sampling-linear-speedup}(linear speedup) 
    If \Cref{cond:sampling} holds  with parameter $\theta$, then for any $\epsilon>0$, \Cref{alg:BP-simulation} terminates within $O\left(\theta\cdot T+\log \left(\frac{n}{\epsilon}\right)\right)$ iterations with probability at least $1-\epsilon$.
    \item \label{lemma:BP-fast-sampling-exp-speedup}(exponential speedup) 
    If \Cref{cond:sampling} holds  with parameter $\theta<1$, then for any $\epsilon>0$, \Cref{alg:BP-simulation} terminates within $O\left(\frac{1}{1-\theta}\log\left(\frac{nT}{\epsilon}\right)\right)$ iterations with probability at least $1-\epsilon$.
\end{enumerate}
\end{lemma}

To see that \Cref{lemma:BP-fast} and \Cref{lemma:BP-fast2} are indeed consequences of \Cref{lemma:BP-fast-sampling},  observe that \Cref{cond:main} with parameter $\rho$, together with an $\alpha$-competitive $\Sample(\cdot,\cdot)$ subroutine, implies \Cref{cond:sampling} with parameter $\theta=\alpha\rho$. 

We now prove \Cref{lemma:BP-fast-sampling}.
Fix an update schedule $\mathfrak{T}=(t_i^v)_{v\in V,0\le i\le m_v}$. 
This defines a directed acyclic graph $D_\prec$ as in~\eqref{dag-def}.
For a \emph{chain $\pi$ of size} $\ell$ in $D_\prec$:
\begin{align*}
\pi: (v_1,i_1) \prec (v_2,i_2) \prec \cdots \prec (v_\ell,i_\ell),
\end{align*}
define its \emph{weight} as follows: For $\ell=1$  set $w(\pi)\triangleq1$; and for $\ell>1$, define
\begin{align}
\label{influence-func}
w(\pi) \triangleq \prod_{i=1}^{\ell-1}\SampleMat({v_i},{v_{i+1}}).
\end{align}

\begin{lemma}
\label{lem:all-path-sum}
Fix an arbitrary update schedule $\mathfrak{T}=(t_i^v)_{v\in V,0\le i\le m_v}$.
For any update $(v,i)$, 
and for any $\ell\geq 1$, we have
\begin{align}
    \Pr\left[ \CurConfig{\widehat{X}}{\ell}{v}[i] \neq \CurConfig{\widehat{X}}{\ell-1}{v}[i] \right] \leq 
    \sum_{\substack{\text{chain $\pi$ of size $\ell$}\\\text{that ends at }(v,i)
    }}w(\pi),\label{eq:branching-upper-bound}
\end{align}
where the probability is taken over the random $\mathfrak{R}=(\RandSeed_i^v)_{v\in V,1\le i\le m_v}$.
\end{lemma}

\begin{proof}
The lemma is proved by establishing the following inequality for all $\ell\ge 2$:
\begin{align}
 \Pr\left[ \CurConfig{\widehat{X}}{\ell}{v}[i] \neq \CurConfig{\widehat{X}}{\ell-1}{v}[i] \right] 
&\leq \sum_{(u,j):(u,j)\prec(v,i)}  \SampleMat({u},{v})\cdot\Pr\left[ \CurConfig{\widehat{X}}{\ell-1}{u}[j] \neq \CurConfig{\widehat{X}}{\ell-2}{u}[j] \right],\label{eq:branching-recursion}
\end{align}
where the sum is taken over all updates $(u,j)$ in $\mathfrak{T}$, where $u\in V$, $1\le j\le m_u$ such that $(u,j)\prec (v,i)$.

The lemma follows from \eqref{eq:branching-recursion} by induction.
Induction Basis: For $\ell=1$,
the right hand side of \eqref{eq:branching-upper-bound} becomes $w((v,i))=1$. Thus,  $\Pr\left[ \CurConfig{\widehat{X}}{1}{v}[i] \neq \CurConfig{\widehat{X}}{0}{v}[i] \right]\le 1$, and hence \eqref{eq:branching-upper-bound} holds.
Induction Step:
Assume that \eqref{eq:branching-upper-bound} holds for chains of size $\ell-1$. To prove \eqref{eq:branching-upper-bound} for chains of size $\ell$, consider \eqref{eq:branching-recursion}. We have:
\begin{align*}
\Pr\left[ \CurConfig{\widehat{X}}{\ell}{v}[i] \neq \CurConfig{\widehat{X}}{\ell-1}{v}[i] \right]
&\le 
\sum_{(u,j):(u,j)\prec(v,i)}  \SampleMat({u},{v}) \sum_{\substack{\text{chain $\pi$ of size $\ell-1$}\\\text{that ends at }(u,j)}}w(\pi)\\
&\le
 \sum_{\substack{\text{chain $\pi$ of size $\ell$}\\\text{that ends at }(v,i)}}w(\pi).
\end{align*}
This completes the induction and proves the lemma.
It remains to prove \eqref{eq:branching-recursion}.

For any $\ell\ge 1$, the variable $\CurConfig{\widehat{X}}{\ell}{v}[i]$ is updated in \Cref{alg:line:BP-update} of \Cref{alg:BP-simulation}.
 Specifically, it is assigned as:
\[
\CurConfig{\widehat{X}}{\ell}{v}[i]\gets \Sample\left(P_v^{\tau^{(\ell)}},\RandSeed_{(v,i)}\right),
\]
where $\tau^{(\ell)}\in Q^{N_v^+}$ is constructed such that for each $u\in N_v^+$:
\begin{align}
\tau^{(\ell)}_u = \CurConfig{\widehat{X}}{\ell-1}{u}[j_u], \quad\text{ where }j_u\triangleq \max\{j\ge 0\mid t_j^u<t_i^v\}.\label{eq:neighborhood-construction}
\end{align}
This means that $\tau^{(\ell)}$ is constructed based on the most recent update time before $t_i^v$ for each neighbor $u$ of $v$.

We use $\mathfrak{R}_{<t}$ to denote the random bits used for resolving updates before time $t$. Formally, we define:
\[
\mathfrak{R}_{<t}\triangleq(\RandSeed_{(v,i)})_{v\in V,1\le i\le m_v,t_{i}^v<t}.
\]
Observe that  for any $\ell\ge 0$, the value of $\CurConfig{\widehat{X}}{\ell}{v}[i]$ is fully determined by $\mathfrak{R}_{<t}$ if $t_i^v<t$.

Let $\tau^{(\ell)}$ and $\tau^{(\ell-1)}$ be constructed as in \eqref{eq:neighborhood-construction}.
Both $\tau^{(\ell)}$ and $\tau^{(\ell-1)}$ 
are determined by $\mathfrak{R}_{<t_i^v}$ and are independent of the random choice of $\RandSeed_{(v,i)}$.
Consider the vertices in $N_v^+$, denoted as $u_1,u_2,\ldots,u_k$, where $k=|N_v^+|$.
We define a sequence of configurations  $\tau_0,\tau_1,\ldots,\tau_k\in Q^{N_v^+}$
that transitions  from $\tau_0=\tau^{(\ell-1)}$ to $\tau_k=\tau^{(\ell)}$  by modifying one spin at a time. 
Specifically, for each $1\leq i\leq k$, 
$\tau_i$ is obtained by changing the spin of $u_i$ in $\tau_{i-1}$ from $\tau^{(\ell-1)}_{u_i}$ to $\tau^{(\ell)}_{u_i}$.
Note that $\tau_{i-1}$ and $\tau_i$ might identical if $\tau^{(\ell-1)}_{u_i}=\tau^{(\ell)}_{u_i}$.

By the triangle inequality, we have: 
\begin{align}
&\Pr\left[ \CurConfig{\widehat{X}}{\ell}{v}[i] \neq \CurConfig{\widehat{X}}{\ell-1}{v}[i] \,\Big{|}\, \mathfrak{R}_{<t_i^v}\right] \notag\\
= &
\Pr\left[ \Sample\left(P_v^{\tau^{(\ell)}},\RandSeed_{(v,i)}\right)\neq \Sample\left(P_v^{\tau^{(\ell-1)}},\RandSeed_{(v,i)}\right) \,\Big{|}\, \mathfrak{R}_{<t_i^v}\right]\notag\\
\le &
\sum_{i=1}^k\Pr\left[ \Sample\left(P_v^{\tau_{i-1}},\RandSeed_{(v,i)}\right)\neq \Sample\left(P_v^{\tau_{i}},\RandSeed_{(v,i)}\right) \,\Big{|}\, \mathfrak{R}_{<t_i^v}\right]\notag\\
= &
\sum_{i=1}^k
\Pr\left[ \tau_{i-1} \neq \tau_i \,\Big{|}\, \mathfrak{R}_{<t_i^v}\right]
\Pr\left[ \Sample\left(P_v^{\tau^{(\ell)}},\RandSeed_{(v,i)}\right)\neq \Sample\left(P_v^{\tau^{(\ell-1)}},\RandSeed_{(v,i)}\right) \,\Big{|}\, \tau_{i-1} \neq \tau_i\right]
\notag\\
\le &
\sum_{i=1}^k \Pr\left[ \tau_{i-1} \neq \tau_i \,\Big{|}\, \mathfrak{R}_{<t_i^v}\right]\cdot \SampleMat(u_i,v) \notag\\
= & \sum_{u\in N^+_v}\Pr\left[\CurConfig{\widehat{X}}{\ell-1}{u}[j_u]\neq \CurConfig{\widehat{X}}{\ell-2}{u}[j_u]\,\Big{|}\, \mathfrak{R}_{<t_i^v}\right]\cdot \SampleMat(u,v),
\label{eq:branching-recursion-2}
\end{align}
where the last inequality follows because the sequence
$\tau^{(\ell-1)}=\tau_0,\tau_1,\ldots,\tau_k=\tau^{(\ell)}$ is constructed such that
$\tau_{i-1}$ and $\tau_i \in Q^{N^+_v}$ differ only at site $u_i$
and by \Cref{def:sampling-matrix}, the probability that $\Sample\left(P_v^{\tau_{i}},\RandSeed_{(v,i)}\right)$ differs from $\Sample\left(P_v^{\tau_{i-1}},\RandSeed_{(v,i)}\right)$ is upper bounded by $\SampleMat(u_i,v)$ 
for any pair $\tau_{i-1},\tau_i$ that differ only at site $u_i$.

Observe that in \eqref{eq:branching-recursion-2},
the sum is actually taken over all updates $(u,j_u)$ such that $(u,j_u)\prec(v,i)$ (with $j_u\ge 1$, since otherwise $\CurConfig{\widehat{X}}{\ell-1}{u}[0]= \CurConfig{\widehat{X}}{\ell-2}{u}[0]=X_0(u)$).
Therefore, we have
\begin{align}
\Pr\left[ \CurConfig{\widehat{X}}{\ell}{v}[i] \neq \CurConfig{\widehat{X}}{\ell-1}{v}[i] \,\Big{|}\, \mathfrak{R}_{<t_i^v}\right] \le \sum_{\substack{(u,j):(u,j)\prec (v,i)}}\SampleMat(u,v)\cdot\Pr\left[\CurConfig{\widehat{X}}{\ell-1}{u}[j]\neq \CurConfig{\widehat{X}}{\ell-2}{u}[j]\,\Big{|}\, \mathfrak{R}_{<t_i^v}\right]. \notag
\end{align}
Then \eqref{eq:branching-recursion} follows by averaging over all $\mathfrak{R}_{<t_i^v}$ using the law of total probability.
\end{proof}

Using \Cref{lem:all-path-sum}, we can prove \Cref{lemma:BP-fast-sampling}, which provides upper bounds on the fast convergence of  \Cref{alg:BP-simulation}, by bounding the expected total weights for all chains of fixed size.

\begin{proof}[Proofs of \Cref{lemma:BP-fast-sampling}]
Let $\mathfrak{T}=(T_i^v)_{v\in V,0\le i\le M_v}$ be a random update schedule generated by $n$ independent rate-1 Poisson clocks  up to time $T$.
Consider the binary relation $\prec$ defined in~\eqref{order-def} and the dependency graph $D_{\prec}$ defined in~\eqref{dag-def}, both derived from this $\mathfrak{T}$.  

Due to \Cref{lem:all-path-sum}, by applying a union bound over all updates $(v,i)$ and averaging over $\mathfrak{T}$, we have
\begin{align}
\Pr\left[\,\widehat{X}^{(\ell)}\neq \widehat{X}^{(\ell-1)}\,\right]
\le 
\mathbb{E}_{\mathfrak{T}}\left[\,\sum_{\text{chain $\pi$ of size $\ell$}}w(\pi)\,\right],\label{eq:BP-fast-1}
\end{align}
where the weight function $w(\cdot)$ is defined in \eqref{influence-func}.

It remains to upper bound the expected total weight referred to in \eqref{eq:BP-fast-1}. 
Let $M=\sum_{v\in V}M_v$ denote the total number of updates in the schedule $\mathfrak{T}=(T_i^v)_{v\in V,0\le i\le M_v}$.
It follows the Poisson distribution with mean $nT$, which is given by:
\[
\forall m\ge 0,\quad  \Pr[M=m]=\mathrm{e}^{-nT}\frac{(nT)^m}{m!}.
\]

Conditioning on $M=m$ for a fixed integer $m\ge 1$, we can sort all $m$ updates according to their respective times as: $0<T_{I_1}^{U_1}<T_{I_2}^{U_2}<\cdots<T_{I_m}^{U_m}<T$. This defines a random sequence of updates:
\begin{align*}
(U_1,\ID_1),(U_2,\ID_2),\ldots,(U_m,\ID_m)\in V\times \{1,2,\ldots,m\}, 
\end{align*}
where $(U_i,\ID_i)$ represents the $i$-th update  in global time order.
Due to the independence and symmetry of the Poisson clocks,  the sequence $\boldsymbol{U}=(U_1,\ldots,U_m)\in V^m$ is uniformly distributed.

Given a sequence $\boldsymbol{U}=(U_1,\ldots,U_m)\in V^m$ of fixed length $m$, 
for any increasing subsequence $1\le j_1<\cdots<j_\ell\le m$ with $1\le \ell\le m$, define:
\begin{equation}\label{eq:weight-function-vertices}
\begin{aligned}
W(j_1,\ldots,j_\ell)=W_{\boldsymbol{U}}(j_1,\ldots,j_\ell) \triangleq \prod_{i=1}^{\ell}\SampleMat\left({U_{j_i}},{U_{j_{i+1}}}\right). 
\end{aligned}
\end{equation}

For a random schedule $\mathfrak{T}=(T_i^v)_{v\in V,0\le i\le M_v}$ conditioned on $M=m$, and the sequence $\boldsymbol{U}=(U_1,\ldots,U_m)$ of updated sites determined by $\mathfrak{T}$,
it is straightforward to verify that for every $1\le \ell \le m$,
\begin{align}
\sum_{\text{chain $\pi$ of size $\ell$}}w(\pi)
\leq
\sum_{1\leq j_1<j_2<\cdots<j_{\ell}\leq m}W_{\boldsymbol{U}}(j_1,\ldots,j_\ell).\label{eq:BP-fast-2}
\end{align}
Let $\boldsymbol{U}=(U_1,U_2,\ldots,U_m)\in V^m$ be chosen uniformly at random.
For any $1\le j_1<\cdots<j_\ell\le m$ and any $v_1,\ldots,v_\ell\in V$, the probability that $(U_{j_1},\ldots,U_{j_\ell})=(v_1,\ldots,v_\ell)$ is $\frac{1}{n^\ell}$. 
Thus, for any $\frac{1}{p}+\frac{1}{q}=1$,
\begin{align}
{\mathbb{E}}\left[W_{\boldsymbol{U}}(j_1,\ldots,j_\ell)\right]
&=
\frac{1}{n^\ell}\sum_{v_1,\ldots,v_\ell\in V}\prod_{i=1}^{\ell-1}\SampleMat(v_{i},v_{i+1}) 
=
\frac{1}{n^\ell}\left\|\SampleMat^{\ell-1}\boldsymbol{1}\right\|_1\nonumber \\
\text{(H\"{o}lder's inequality)}\qquad
&\le
\frac{n^{1/q}}{n^\ell}\left\|\SampleMat^{\ell-1}\boldsymbol{1}\right\|_p\nonumber \\
&\le 
\frac{n^{1/q}}{n^{\ell}}\left\|\SampleMat\right\|_p^{\ell-1}\|\boldsymbol{1}\|_p \nonumber \\
&=
\frac{1}{n^{\ell-1}}\left\|\SampleMat\right\|_p^{\ell-1}.
\label{eq:Holder-inequality}
\end{align}
By \Cref{cond:main}, it holds that $\|\SampleMat\|_p\le \theta$ for some $p\in[1,\infty]$.
Hence 
$\mathbb{E}\left[W_{\boldsymbol{U}}(j_1,\ldots,j_\ell)\right]\le \left(\frac{\theta}{n}\right)^{\ell-1}$.
Combining with \eqref{eq:BP-fast-2}, for every $1\le \ell\le m$,
\begin{align}
{\mathbb{E}}_{\mathfrak{T}}\left[\,\sum_{\text{chain $\pi$ of size $\ell$}}w(\pi)\,\Biggm{|}\,M=m\,\right]
&\le
\mathop{\mathbb{E}}_{\boldsymbol{U}\in V^m}\left[\,\sum_{1\le j_1<\cdots<j_\ell\le m}W_{\boldsymbol{U}}(j_1,\ldots,j_\ell)\,\right]  \nonumber  \\
&\le 
\binom{m}{\ell} \left(\frac{\theta}{n}\right)^{\ell-1}.
\label{eq:eq:converge-1}
\end{align}
Combined with \eqref{eq:BP-fast-1},
for any $\ell\ge 1$,  
the probability that \Cref{alg:BP-simulation} does not terminate within $\ell\ge 1$ iterations is:

\begin{align*}
\Pr\left[\,\widehat{X}^{(\ell)}\neq \widehat{X}^{(\ell-1)}\,\right]
&\le
\sum_{m\ge 0}\Pr[M=m]\cdot \mathbb{E}_{\mathfrak{T}}\left[\,\sum_{\text{chain $\pi$ of size $\ell$}}w(\pi)\,\Biggm{|}\,M=m\,\right]\\
&\le 
\sum_{m\ge \ell}
\mathrm{e}^{-nT}\frac{(nT)^m}{m!}{m\choose \ell}\left(\frac{\theta}{n}\right)^{\ell-1}\\
\text{(take $k=m-\ell$)}\qquad
&=
\frac{nT}{\ell!}(\theta T)^{\ell-1}
\sum_{k\ge 0}\mathrm{e}^{-nT}\frac{(nT)^k}{k!}\\
\mbox{($(\ell-1)!\ge\left(\frac{\ell-1}{\mathrm{e}}\right)^{\ell-1}$)}\qquad
&\le
\frac{nT}{\ell} \left(\frac{\mathrm{e}\theta T}{\ell-1}\right)^{\ell-1}.
\end{align*}
To ensure this is at most $\epsilon$,
set
$\ell=\left\lceil(2\mathrm{e}\theta+1)T+\log_2\left(\frac{n}{\epsilon}\right)+1\right\rceil$.
This proves \Cref{lemma:BP-fast-sampling-linear-speedup} of \Cref{lemma:BP-fast-sampling}.

Next, to prove \Cref{lemma:BP-fast-sampling-exp-speedup} of \Cref{lemma:BP-fast-sampling}, consider a more refined upper bound than \eqref{eq:BP-fast-2}:
\begin{align}
\sum_{\text{chain $\pi$ of size $\ell$}}w(\pi)
\leq 
\sum_{1\leq j_1<\cdots<j_{\ell}\leq m}I\left[(U_{j_1},\ID_{j_1}) \prec \cdots \prec (U_{j_\ell},\ID_{j_\ell})\right]\cdot W_{\boldsymbol{U}}(j_1,\ldots,j_\ell).\label{eq:BP-fast-2-stricter}
\end{align}
Thus, for every $1\le \ell\le m$, we have
\begin{align}
& {\mathbb{E}}_{\mathfrak{T}}\left[\,\sum_{\text{chain $\pi$ of size $\ell$}}w(\pi)\,\Biggm{|}\,M=m\,\right] \notag \\
\le &
\mathop{\mathbb{E}}_{\boldsymbol{U}\in V^m}\left[\,\sum_{1\le j_1<\cdots<j_\ell\le m}I\left[ (U_{j_1},\ID_{j_1})\prec\cdots\prec(U_{j_\ell},\ID_{j_\ell}) \right]W_{\boldsymbol{U}}(j_1,\ldots,j_\ell)\,\right]  \notag  \\
= &
\sum_{1\leq j_1<\cdots<j_\ell\leq m} \sum_{(v_1,\ldots,v_\ell)\in V^{\ell}} \Pr_{\boldsymbol{U}\in V^m}\left[ (U_{j_1},\ID_{j_1})\prec\cdots\prec(U_{j_\ell},\ID_{j_\ell}) \wedge (U_{j_1},\ldots,U_{j_\ell})=(v_1,\ldots,v_\ell) \right] \notag \\
& \quad  \cdot \mathop{\mathbb{E}}_{\boldsymbol{U}\in V^m}\left[\, W_{\boldsymbol{U}}(j_1,\ldots,j_\ell) \mid (U_{j_1},\ID_{j_1})\prec\cdots\prec(U_{j_\ell},\ID_{j_\ell}) \wedge (U_{j_1},\ldots,U_{j_\ell})=(v_1,\ldots,v_\ell) \,\right]  \notag \\
= &
\sum_{1\leq j_1<\cdots<j_\ell\leq m} \sum_{(v_1,\ldots,v_\ell)\in V^{\ell}} 
\Pr_{\boldsymbol{U}\in V^m}\left[ (U_{j_1},\ID_{j_1})\prec\cdots\prec(U_{j_\ell},\ID_{j_\ell}) \mid (U_{j_1},\ldots,U_{j_\ell})=(v_1,\ldots,v_\ell) \right] \notag \\
& \quad \cdot \Pr_{\boldsymbol{U}\in V^m}\left[ (U_{j_1},\ldots,U_{j_\ell})=(v_1,\ldots,v_\ell) \right] \cdot  W_{(v_1,\ldots,v_\ell)}(j_1,\ldots,j_\ell).
\label{eq:eq:converge-2}
\end{align}
Here, the inequality follows from \eqref{eq:BP-fast-2-stricter}, and the equations are due to the total expectation and total probability.

Let $\boldsymbol{U}=(U_1,U_2,\ldots,U_m)\in V^m$ be chosen uniformly at random.
For any $1\leq i<j\leq m$, we have $(U_i,\ID_i)\prec (U_j,\ID_j)$ only if $U_k\neq U_i$ for all $i<k<j$. Therefore, for any $1\leq j_1<\cdots<j_{\ell}\leq m$ and $(v_1,v_2,\ldots,v_{\ell})\in V^{\ell}$, we have
\begin{align}
    \Pr\left[(U_{j_1},\ID_{j_1})\prec\cdots\prec(U_{j_\ell},\ID_{j_\ell}) \biggm{|} (U_{j_1},\ldots,U_{j_\ell})=(v_1,\ldots,v_\ell)\right] 
    &\leq \prod_{i=1}^{\ell-1}\left(1-\frac{1}{n}\right)^{j_{i+1}-j_i-1} 
    = \left(1-\frac{1}{n}\right)^{j_\ell-j_1-(\ell-1)}. \label{eq:dependent-prob}
\end{align}

By \cref{cond:main}, combined with \eqref{eq:Holder-inequality}, \eqref{eq:eq:converge-2} and \eqref{eq:dependent-prob}, we obtain that for every $1\le \ell\le m$, 
\begin{align*}
& {\mathbb{E}}_{\mathfrak{T}}\left[\,\sum_{\text{chain $\pi$ of size $\ell$}}w(\pi)\,\Biggm{|}\,M=m\,\right]  \\
\le &
\sum_{1\leq j_1<\cdots<j_\ell\leq m} \sum_{(v_1,\ldots,v_\ell)\in V^{\ell}} 
\Pr_{\boldsymbol{U}\in V^m}\left[ (U_{j_1},\ID_{j_1})\prec\cdots\prec(U_{j_\ell},\ID_{j_\ell}) \mid (U_{j_1},\ldots,U_{j_\ell})=(v_1,\ldots,v_\ell) \right] \\
& \quad \cdot \Pr_{\boldsymbol{U}\in V^m}\left[ (U_{j_1},\ldots,U_{j_\ell})=(v_1,\ldots,v_\ell) \right] \cdot  W_{(v_1,\ldots,v_\ell)}(j_1,\ldots,j_\ell) \\
\le &
\sum_{1\leq j_1<\cdots<j_\ell\leq m} \left(1-\frac{1}{n}\right)^{j_\ell-j_1-(\ell-1)} \sum_{(v_1,\ldots,v_\ell)\in V^{\ell}} \Pr_{\boldsymbol{U}\in V^m}\left[ (U_{j_1},\ldots,U_{j_\ell})=(v_1,\ldots,v_\ell) \right] \cdot  W_{(v_1,\ldots,v_\ell)}(j_1,\ldots,j_\ell)  \\
= &
\sum_{1\leq j_1<\cdots<j_\ell\leq m} \left(1-\frac{1}{n}\right)^{j_\ell-j_1-(\ell-1)} \mathop{\mathbb{E}}_{\boldsymbol{U}\in V^m}\left[\, W_{\boldsymbol{U}}(j_1,\ldots,j_\ell)\,\right]   \\
\le &
\sum_{1\leq j_1<\cdots<j_\ell\leq m} \left(1-\frac{1}{n}\right)^{j_\ell-j_1-(\ell-1)}\left(\frac{\|\SampleMat\|_p}{n}\right)^{\ell-1}\\
\le &
\sum_{1\leq j_1<\cdots<j_\ell\leq m} \left(1-\frac{1}{n}\right)^{j_\ell-j_1-(\ell-1)}\left(\frac{\theta}{n}\right)^{\ell-1}.
\end{align*}
If $\theta<1$, then combined with \eqref{eq:BP-fast-1}, we have
\begin{align*}
\Pr\left[\,\widehat{X}^{(\ell)}\neq \widehat{X}^{(\ell-1)}\,\right]
&\le
\sum_{m\ge 0}\Pr[M=m]\cdot \mathbb{E}_{\mathfrak{T}}\left[\,\sum_{\text{chain $\pi$ of size $\ell$}}w(\pi)\,\Biggm{|}\,M=m\,\right]\\
&\le 
\sum_{m\ge 0}\Pr[M=m]\cdot \sum_{1\leq j_1<\cdots<j_\ell\leq m} \left(1-\frac{1}{n}\right)^{j_\ell-j_1-(\ell-1)}\left(\frac{\theta}{n}\right)^{\ell-1}\\
\text{(take $k_i=j_{i+1}-j_i-1$)}\qquad &\le 
\sum_{m\ge 0}\Pr[M=m]\cdot \sum_{1\leq j_1\leq m} \sum_{k_1,k_2,\ldots,k_{\ell-1}\geq 0} \left(1-\frac{1}{n}\right)^{\sum_{i=1}^{\ell-1}k_i} \left(\frac{\theta}{n}\right)^{\ell-1} \\
&\le
\sum_{m\ge 0}\Pr[M=m]\cdot mn^{\ell-1}\left(\frac{\theta}{n}\right)^{\ell-1} \\
&= \mathbb{E}[m] \theta^{\ell-1}\\
&= nT \theta^{\ell-1}.
\end{align*}
Taking $\ell = \left\lceil \frac{1}{1-\theta}\ln\left(\frac{nT}{\epsilon}\right) \right\rceil+1$, the probability that \Cref{alg:BP-simulation} does not terminate within $\ell\ge 1$ iterations is at most $\epsilon$. 
This proves \Cref{lemma:BP-fast-sampling-exp-speedup} of \Cref{lemma:BP-fast-sampling}.
\end{proof}


\section{Universal Coupling via Correlated Sampling} \label{sec:universal-coupling}
In this section, we present the {universal coupling} claimed in \Cref{thm:2-competitive}.

Let $\Delta(\Omega)$ denote the \emph{probability simplex} on the sample space $\Omega$. Formally,
\[
\Delta(\Omega)\triangleq \left\{p\in[0,1]^\Omega \bigm{|} \mbox{$\sum_{x\in\Omega}p(x)=1$}\right\}.
\]
Each $p\in\Delta(\Omega)$ represents a probability distribution over $\Omega$, where the probability of $x\in \Omega$ is $p(x)$.

\begin{definition}[{universal coupling}]\label{def:universal-coupling}
A deterministic function $\Sample:\Delta(\Omega)\times [0,1]\to \Omega$ is a \emph{universal coupling} on sample space $\Omega$ 
if, 
when $\RandSeed\in[0,1]$ is chosen uniformly at random,
for any distribution $p\in\Delta(\Omega)$ and $x\in \Omega$, 
\[
\Pr_{\RandSeed}\left[\,\Sample(p,\RandSeed)=x\,\right]=p(x).
\]
\end{definition}

A universal coupling simultaneously couples all distributions over $\Omega$. 
Our objective is to achieve a universal coupling with small competitiveness (as defined in \Cref{def:competitive-sample}), meaning that for any pair of distributions $p,q\in\Delta(\Omega)$, the probability $\Pr[\Sample(p,\RandSeed)\neq \Sample(q,\RandSeed)]$ should be as close as possible to the best achievable value $\dtv(p,q)$.

Such a universal coupling with small competitiveness can be achieved by the following sampling procedure.
Note that the uniform random $\RandSeed\in[0,1]$ is mapped to an infinite sequence $(X_1,Y_1),(X_2,Y_2),\ldots$, where each $(X_i,Y_i)\in\Omega\times[0,1]$ is chosen uniformly and independently at random. 
\begin{algorithm}[ht]
 \SetKwInOut{Input}{Input} \SetKwInOut{Output}{Output} 
\Input{ $p\in\Delta(\Omega)$; $\RandSeed=((X_1,Y_1),(X_2,Y_2),\ldots)$, where each $(X_i,Y_i)\in\Omega\times[0,1]$.}
\Output{ a sample in $\Omega$.}
let $i^*$ be the smallest $i\ge 1$ such that $Y_i<p(X_i)$\;
\Return{$X_{i^*}$}\; 
\caption{$\Sample(p,\RandSeed)$\label{alg:sample}}
\end{algorithm}

\begin{remark}
This sampling procedure has been independently discovered in various contexts to solve different problems.
Notably, it was identified by Kleinberg and Tardos~\cite{kleinberg2002approximation} for rounding linear programs, 
and by Holenstein~\cite{holenstein2007parallel} to simplify Raz's proof of parallel repetition theorem.
As Charikar~\cite{charikar2002similarity} pointed out, this sampling strategy generalizes Broder's MinHash strategy~\cite{broder1997resemblance} for universal coupling of uniform distributions.
The optimality of such sampling strategies was further studied in~\cite{bavarian2020optimality}.
In that work and its follow-up studies, such as~\cite{bun2023stability,ghazi2021user,sudan2019communication,karbasi2023replicability,kalavasis2023statistical,angel2019pairwise}, sampling strategies satisfying \Cref{def:universal-coupling} were referred to as \emph{correlated sampling} strategies.
\end{remark}

The correctness, efficiency, and competitiveness of this sampling procedure are ensured by the following lemmas:
\begin{lemma}[{correctness and efficiency}]\label{lemma:sample-correctness-efficiency}
For any $p\in\Delta(\Omega)$, with random $\RandSeed=((X_1,Y_1),(X_2,Y_2),\ldots)$ 
where each $(X_i,Y_i)\in \Omega\times[0,1]$ is chosen uniformly and independently at random:
\begin{enumerate}
\item\label{sample-correctness} For any $x\in\Omega$, the sampling procedure returns $x$ with probability $p(x)$, i.e., $\Pr\left[\,\Sample(p,\RandSeed)=x\,\right]=p(x)$;
\item\label{sample-efficiency}  The index $i^*$ chosen in \Cref{alg:sample} follows a geometric distribution a success probability of $1/|\Omega|$.
\end{enumerate}
\end{lemma}

\begin{lemma}[{coupling performance}]\label{lemma:sample-coupling-performance} 
For any $p,q\in\Delta(\Omega)$, 
with random $\RandSeed=((X_1,Y_1),(X_2,Y_2),\ldots)$ where each $(X_i,Y_i)\in \Omega\times[0,1]$ is chosen uniformly and independently at random:
\begin{align}
\label{eq:universal-couple-coalesce}
\Pr\left[\,\Sample(p,\RandSeed)= \Sample(q,\RandSeed)\,\right]\ge\frac{1-\dtv(p,q)}{1+\dtv(p,q)}.
\end{align}
\end{lemma}

\Cref{thm:2-competitive} follows immediately from  \cref{lemma:sample-correctness-efficiency} and \cref{lemma:sample-coupling-performance}, 
because  it holds that 
\[
\Pr\left[\,\Sample(p,\RandSeed)\neq \Sample(q,\RandSeed)\,\right]\le\frac{2\dtv(p,q)}{1+\dtv(p,q)}\le 2\dtv(p,q).
\]

\begin{remark}
The coupled probability in \eqref{eq:universal-couple-coalesce} corresponds to the following \emph{Jaccard similarity}: 
\[
\frac{1-\dtv(p,q)}{1+\dtv(p,q)}
=
\frac{\sum_{x\in \Omega}\min(p(x),q(x))}{\sum_{x\in \Omega}\max(p(x),q(x))}
=
\frac{\left|P\cap Q\right|}{\left|P\cup Q\right|}.
\]
This measures the similarity between two points $P,Q \subseteq \Omega \times [0,1]$ 
naturally constructed from the distributions $p,q\in\Delta(\Omega)$, respectively,
as:
$P\triangleq\left\{(x,y)\mid x\in\Omega\land y< p(x)\right\}$
and $Q\triangleq\left\{(x,y)\mid x\in\Omega\land y< q(x)\right\}$.
An illustration of these point sets $P$ and $Q$  is provided in the following figure.

\definecolor{color1}{RGB}{180,30,0}
\definecolor{color2}{RGB}{0,90,180}

\begin{center}
\begin{tikzpicture}[scale=0.77]
\draw[line width=1.2pt, ->] (0,0) -- (10,0); 
\draw[line width=1.2pt, ->] (0,0) -- (0,5); 
\node[anchor=mid] at (4.900,-0.500) {$x$ \tikz[baseline]{(0,0)--(.10,0);}};
\node[anchor=mid] at (-0.350,5.000) {$\mu$ \tikz[baseline]{(0,0)--(.10,0);}};

\node[anchor=mid] at (3.660,0.900) {$P \cap Q$ \tikz[baseline]{(0,0)--(.10,0);}};
\node[anchor=mid] at (1.360,2.600) {$P \setminus Q$ \tikz[baseline]{(0,0)--(.10,0);}};
\node[anchor=mid] at (6.760,2.200) {$Q \setminus P$ \tikz[baseline]{(0,0)--(.10,0);}};

\draw [opacity=0.0, fill=color1, fill opacity=0.18, smooth, samples=50]
(0,0)
-- (0.000, 4.800)
-- (0.900, 4.800)
-- (0.900, 4.447)
-- (1.800, 4.447)
-- (1.800, 4.424)
-- (2.700, 4.424)
-- (2.700, 2.682)
-- (3.600, 2.682)
-- (3.600, 2.682)
-- (4.500, 2.682)
-- (4.500, 0.871)
-- (5.400, 0.871)
-- (5.400, 1.059)
-- (6.300, 1.059)
-- (6.300, 0.941)
-- (7.200, 0.941)
-- (7.200, 0.824)
-- (8.100, 0.824)
-- (8.100, 0.800)
-- (9.000, 0.800)
-- (9.000, 0.000)
-- cycle;

\draw [opacity=0.0, fill=color2, fill opacity=0.18, smooth, samples=50]
(0,0)
-- (0.000, 0.918)
-- (0.900, 0.918)
-- (0.900, 0.424)
-- (1.800, 0.424)
-- (1.800, 0.871)
-- (2.700, 0.871)
-- (2.700, 1.506)
-- (3.600, 1.506)
-- (3.600, 3.459)
-- (4.500, 3.459)
-- (4.500, 3.200)
-- (5.400, 3.200)
-- (5.400, 3.224)
-- (6.300, 3.224)
-- (6.300, 3.388)
-- (7.200, 3.388)
-- (7.200, 2.612)
-- (8.100, 2.612)
-- (8.100, 3.929)
-- (9.000, 3.929)
-- (9.000, 0.000)
-- cycle;
    
\color{color1}
\draw [line width=1.3pt] (0.000, 4.800) -- (0.900, 4.800);
\draw [line width=1.3pt] (0.900, 4.447) -- (1.800, 4.447);
\draw [line width=1.3pt] (1.800, 4.424) -- (2.700, 4.424);
\draw [line width=1.3pt] (2.700, 2.682) -- (3.600, 2.682);
\draw [line width=1.3pt] (3.600, 2.682) -- (4.500, 2.682);
\draw [line width=1.3pt] (4.500, 0.871) -- (5.400, 0.871);
\draw [line width=1.3pt] (5.400, 1.059) -- (6.300, 1.059);
\draw [line width=1.3pt] (6.300, 0.941) -- (7.200, 0.941);
\draw [line width=1.3pt] (7.200, 0.824) -- (8.100, 0.824);
\draw [line width=1.3pt] (8.100, 0.800) -- (9.000, 0.800);
\draw [line width=1.1pt] (0.900, 4.800) -- (0.900, 4.447);
\draw [line width=1.1pt] (1.800, 4.447) -- (1.800, 4.424);
\draw [line width=1.1pt] (2.700, 4.424) -- (2.700, 2.682);
\draw [line width=1.1pt] (3.600, 2.682) -- (3.600, 2.682);
\draw [line width=1.1pt] (4.500, 2.682) -- (4.500, 0.871);
\draw [line width=1.1pt] (5.400, 0.871) -- (5.400, 1.059);
\draw [line width=1.1pt] (6.300, 1.059) -- (6.300, 0.941);
\draw [line width=1.1pt] (7.200, 0.941) -- (7.200, 0.824);
\draw [line width=1.1pt] (8.100, 0.824) -- (8.100, 0.800);
\node[anchor=mid] at (9.250,0.800) {$p$ \tikz[baseline]{(0,0)--(.10,0);}};

\draw [dashed] (0.900, 0.000) -- (0.900, 4.447);
\draw [dashed] (1.800, 0.000) -- (1.800, 4.424);
\draw [dashed] (2.700, 0.000) -- (2.700, 2.682);
\draw [dashed] (3.600, 0.000) -- (3.600, 2.682);
\draw [dashed] (4.500, 0.000) -- (4.500, 0.871);
\draw [dashed] (5.400, 0.000) -- (5.400, 0.871);
\draw [dashed] (6.300, 0.000) -- (6.300, 0.941);
\draw [dashed] (7.200, 0.000) -- (7.200, 0.824);
\draw [dashed] (8.100, 0.000) -- (8.100, 0.800);

\color{color2}
\draw [line width=1.3pt] (0.000, 0.918) -- (0.900, 0.918);
\draw [line width=1.3pt] (0.900, 0.424) -- (1.800, 0.424);
\draw [line width=1.3pt] (1.800, 0.871) -- (2.700, 0.871);
\draw [line width=1.3pt] (2.700, 1.506) -- (3.600, 1.506);
\draw [line width=1.3pt] (3.600, 3.459) -- (4.500, 3.459);
\draw [line width=1.3pt] (4.500, 3.200) -- (5.400, 3.200);
\draw [line width=1.3pt] (5.400, 3.224) -- (6.300, 3.224);
\draw [line width=1.3pt] (6.300, 3.388) -- (7.200, 3.388);
\draw [line width=1.3pt] (7.200, 2.612) -- (8.100, 2.612);
\draw [line width=1.3pt] (8.100, 3.929) -- (9.000, 3.929);
\draw [line width=1.1pt] (0.900, 0.918) -- (0.900, 0.424);
\draw [line width=1.1pt] (1.800, 0.424) -- (1.800, 0.871);
\draw [line width=1.1pt] (2.700, 0.871) -- (2.700, 1.506);
\draw [line width=1.1pt] (3.600, 1.506) -- (3.600, 3.459);
\draw [line width=1.1pt] (4.500, 3.459) -- (4.500, 3.200);
\draw [line width=1.1pt] (5.400, 3.200) -- (5.400, 3.224);
\draw [line width=1.1pt] (6.300, 3.224) -- (6.300, 3.388);
\draw [line width=1.1pt] (7.200, 3.388) -- (7.200, 2.612);
\draw [line width=1.1pt] (8.100, 2.612) -- (8.100, 3.929);
\node[anchor=mid] at (9.250,3.929) {$q$ \tikz[baseline]{(0,0)--(.10,0);}};

\draw [dashed] (0.900, 0.000) -- (0.900, 0.424);
\draw [dashed] (1.800, 0.000) -- (1.800, 0.424);
\draw [dashed] (2.700, 0.000) -- (2.700, 0.871);
\draw [dashed] (3.600, 0.000) -- (3.600, 1.506);
\draw [dashed] (4.500, 0.000) -- (4.500, 3.200);
\draw [dashed] (5.400, 0.000) -- (5.400, 3.200);
\draw [dashed] (6.300, 0.000) -- (6.300, 3.224);
\draw [dashed] (7.200, 0.000) -- (7.200, 2.612);
\draw [dashed] (8.100, 0.000) -- (8.100, 2.612);
\end{tikzpicture}
\end{center}

A recent study \cite{bavarian2020optimality} has proven that this is essentially optimal for universal coupling.
\end{remark}

\cref{lemma:sample-correctness-efficiency,lemma:sample-coupling-performance}
have been established in multiple independent works~\cite{kleinberg2002approximation,holenstein2007parallel} within their respective contexts. 
For completeness, we provide the formal proofs of these lemmas here.

\begin{proof}[proof of \cref{lemma:sample-correctness-efficiency}]
Let $(X,Y)\in\Omega\times[0,1]$ be chosen uniformly at random.
For any function $w:\Omega\to[0,1]$,
\begin{align}
\Pr[Y < w(X)] 
= \sum_{x\in \Omega}\Pr[X=x\wedge Y<w(x)]
=\frac{1}{|\Omega|}\sum_{x\in \Omega} w(x).\label{eq:sample-success-prob}
\end{align}
In particular, for $p\in\Delta(\Omega)$, we have $\Pr[Y < p(X)]={1}/{|\Omega|}$ since $\sum_{x\in\Omega}p(x)=1$.
Since the $(X_i,Y_i)$ pairs are independent, this proves \Cref{sample-efficiency}  of the lemma: $i^*$ follows a geometric distribution with a success probability of $1/|\Omega|$.

For any $p\in\Delta(\Omega)$ and $x\in\Omega$,
\begin{align*}
    \Pr[X=x \mid Y<p(X)] 
    = \frac{\Pr[X=x \wedge Y<p(x)]}{\Pr[Y<p(X)]}
    = \frac{{p(x)}/{|\Omega|}}{{1}/{|\Omega|}} = p(x),
\end{align*}
where the second equation follows from the fact that $\Pr[Y < p(X)]={1}/{|\Omega|}$.
Since the output of $\Sample(p,\RandSeed)$ follows the distribution of $X$ conditional on $Y<p(X)$,
this proves \Cref{sample-correctness}  of the lemma: $\Sample(p,\RandSeed)$ always returns a correct sample from $p$ for every  $p\in\Delta(\Omega)$.
\end{proof}

\begin{proof}[proof of \cref{lemma:sample-coupling-performance}]
Let $(p\cap q):\Omega\to[0,1]$ and $(p\cup q):\Omega\to[0,1]$ be two functions defined as: 
\begin{align*}
\forall x\in\Omega:\quad
    (p\cap q)(x) = \min\left(p(x),q(x)\right)
\quad\text{ and }\quad
    (p\cup q)(x) = \max\left(p(x),q(x)\right).
\end{align*}
Clearly, $(p\cap q)(x)\le (p\cup q)(x)$, and it is straightforward to verify that
\begin{align*}
\sum_{x\in \Omega} (p\cap q)(x) 
	= 
	1 - \dtv(p,q)
\quad\text{ and }\quad
\sum_{x\in \Omega} (p\cup q)(x) 
	= 
	1 + \dtv(p,q).
\end{align*}
By \eqref{eq:sample-success-prob}, 
for $(X,Y)$ chosen uniformly at random from $\Omega\times[0,1]$, we have 
\begin{align*}
\Pr[Y<(p\cap q)(X)] 
= 
\frac{1 - \dtv(p,q)}{|\Omega|}
\quad\text{ and }\quad
\Pr[Y<(p\cup q)(X)] 
= 
\frac{1 + \dtv(p,q)}{|\Omega|},
\end{align*}
and thus,
\begin{align}
\Pr[Y<(p\cap q)(X) \mid Y<(p\cup q)(X)]
&=
\frac{\Pr[Y<(p\cap q)(X)]}{\Pr[Y<(p\cup q)(X)]}
=\frac{1-\dtv(p,q)}{1+\dtv(p,q)}.\label{eq:Jaccard-quantity-0}
\end{align}
Let $\RandSeed=((X_1,Y_1),(X_2,Y_2),\ldots)$ be an infinite sequence where each $(X_i,Y_i)\in \Omega\times[0,1]$ is chosen uniformly and independently at random.
Define:
\begin{align*}
I_p^*
=
\min\{i\mid Y_i<p(X_i)\}
\quad\text{ and }\quad
I_q^*
=
\min\{i\mid Y_i<q(X_i)\}.
\end{align*}
Observe that 
$\min\left(I_p^*,I_q^*\right)=\min\{i\mid Y_i<(p\cup q)(X_i)\}$, and
\begin{align}
I_p^*= I_q^*
\implies
\Sample(p,\RandSeed)= \Sample(q,\RandSeed).\label{eq:observation-coalesce-implication}
\end{align}
Then, for every $i\ge 1$,
\begin{align}
&\Pr\left[I_p^*=I_q^*=i\mid \min\left(I_p^*,I_q^*\right)=i\right]\notag\\
=
&\Pr\left[Y_i<(p\cap q)(X_i)\land \left(\bigwedge_{j<i}Y_j\ge (p\cup q)(X_j)\right)\biggm{|}  Y_i<(p\cup q)(X_i)\land \left(\bigwedge_{j<i}Y_j\ge (p\cup q)(X_j)\right)\right]\notag\\
=
&\Pr\left[Y_i<(p\cap q)(X_i) \mid Y_i<(p\cup q)(X_i) \right]\label{eq:Jaccard-quantity-0.5}\\
=
&\frac{1-\dtv(p,q)}{1+\dtv(p,q)},\label{eq:Jaccard-quantity-1}
\end{align}
where \eqref{eq:Jaccard-quantity-0.5} follows from the independence of different $(X_i,Y_i)$ pairs and the fact that:
\[
Y_j\ge (p\cup q)(X_j)\implies Y_j\ge (p\cap q)(X_j),
\]
and \eqref{eq:Jaccard-quantity-1} follows from \eqref{eq:Jaccard-quantity-0}.
Therefore, 
\begin{align}
\Pr\left[I_p^*= I_q^*\right]
&=
\sum_{i\ge 1}\Pr\left[\min\left(I_p^*,I_q^*\right)=i\right]\cdot\Pr\left[I_p^*= I_q^*=i\mid \min\left(I_p^*,I_q^*\right)=i\right]\notag\\
&=
\frac{1-\dtv(p,q)}{1+\dtv(p,q)}\cdot \sum_{i\ge 1}\Pr\left[\min\left(I_p^*,I_q^*\right)=i\right]\label{eq:Jaccard-quantity-2}\\
&=
\frac{1-\dtv(p,q)}{1+\dtv(p,q)}.\notag
\end{align}
where \eqref{eq:Jaccard-quantity-2} follows from \eqref{eq:Jaccard-quantity-1}.
Finally, by the observation in \eqref{eq:observation-coalesce-implication}, 
\[
\Pr\left[\Sample(p,\RandSeed)= \Sample(q,\RandSeed)\right]
\ge 
\Pr\left[I_p^*= I_q^*\right]
=
\frac{1-\dtv(p,q)}{1+\dtv(p,q)}.\qedhere
\]
\end{proof}


\section{Parallel and Distributed Implementations}
In this section, we present implementations of the algorithm for simulating single-site dynamics.

The following theorem provides a formal restatement of \Cref{main-thm:parallel} and \Cref{main-thm:two-spin}.
\begin{theorem}[\Cref{main-thm:parallel} and \ref{main-thm:two-spin}, formally stated]\label{thm:two-spin}\label{thm:parallel}
Assume the existence of oracles for evaluating $\{P_v^{\tau}\}$, 
such that for any $v\in V$, $\tau\in Q^{N_v^+}$ and $x\in Q$, the oracle returns the probability $P_v^{\tau}(x)$.
There exists a \CRAM{} algorithm such that,
given any $X_0\in Q^V$ on graph $G=(V,E)$ and  $0<T\le n^{c_0}$, where $n=|V|$, $m=|E|$ and $q=|Q|$,
the followings hold:
\begin{enumerate}
\item\label{itm:main-parallel-correctness}
The algorithm returns a random configuration $X\in Q^V$ that is identically distributed as $X_T$ in the continuous-time single-site dynamics $(X_t)_{t\in\mathbb{R}_{\ge 0}}$ specified by $\{P_v^{\tau}\}$ on $G$ conditioned on $X_0$.
\item\label{itm:main-parallel-fast}
(linear speedup)
Assume \Cref{cond:main} with parameter $\rho$.
With probability $1-n^{-c_1}$,
the algorithm terminates within $O_{c_0,c_1}\left(\rho\cdot (T+\log n)\right)$ depth 
on $O_{c_1}\left(\left(m+nq^2\log^2 n\right)\log n\right)$ processors. 
\item\label{itm:main-two-spin}
(exponential speedup)
Assume $|Q|=2$ and \Cref{cond:main} with parameter $\rho<1$.
With probability $1-n^{-c_1}$,
the algorithm terminates within $O_{c_1}\left(\frac{1}{1-\rho}(\log T+\log n)\right)$ depth 
on $O_{c_1}\left(mT\right)$ processors.
\end{enumerate} 
\end{theorem}

We first describe the algorithm that works for general finite domain $Q$ and achieves the linear speedup stated in \Cref{itm:main-parallel-fast} of \Cref{thm:two-spin}.
For this case, we use \Cref{alg:sample} on sample space $Q$ as the $\Sample(\cdot,\cdot)$ subroutine used in \Cref{alg:line:BP-update} of \Cref{alg:BP-simulation}.

\begin{algorithm}[ht]
 \SetKwInOut{Input}{Input} \SetKwInOut{Output}{Output} 
\KwData{ integers $k\ge 0$, $w\ge 1$; $\RandSeed=((X_1,Y_1),\ldots,(X_{k w},Y_{k w}))$, where $(X_i,Y_i)\in Q\times[0,1]$.}
  \SetKwFunction{FInit}{$\mathsf{Initialize}$}
  \SetKwProg{Fn}{Function}{:}{}
  \Fn{\FInit{$W$}}{
    \Input{ an integer $W\ge 1$}
        $k\gets 0$, $w\gets W$, and $\RandSeed\gets()$ is initialized to an empty sequence\;
  }
  \SetKwFunction{FDraw}{$\mathsf{Draw}$}
  \Fn{\FDraw{$p$}}{
  \Input{ a distribution $p$ over $Q$.}
\Output{ a random spin in $Q$ distributed as $p$.}
  $\ell\gets 0$\;
  \While{true}{
  \If{$\ell= k$}{
  generate $\left((X_{\ell\cdot w+1},Y_{\ell\cdot w+1}),\ldots,(X_{(\ell+1)\cdot w},Y_{(\ell+1)\cdot w})\right)$, each $(X_i,Y_i)\in Q\times[0,1]$ chosen uniformly and independently at random, and append it to the end of $\RandSeed$\;
  $k\gets k+1$\;
  }
  compute $i^*=\min\left\{i\mid \ell\cdot q+1\le i\le  (\ell+1)\cdot q\land Y_i<p(X_i)\right\}\cup\{\infty\}$\;\tcc{uses $O(1)$ depth and $O(w^2)$ processors in the CRCW PRAM}
  \lIf{$i^*<\infty$}{\Return{$X_{i^*}$}}
  $\ell\gets \ell+1$\;
  }
  }
\caption{Dynamic data structure $\Sampler$\label{alg:sampler-DS}}
\end{algorithm}

\Cref{alg:sample} can be parallelized straightforwardly. 
The sample space is the set $Q$ of all spins.
In each round, 
for given \emph{width} $w\ge 1$,
the algorithm checks in parallel the next $w$ pairs $(X_i,Y_i)\in Q\times[0,1]$ in the sequence $(X_1,Y_1),(X_2,Y_2),\ldots$, to find the pair $(X_i,Y_i)$ with the smallest $i$ that satisfies $Y_i<p(X_i)$ 
(which, in \CRAM{}, uses $O(1)$ depth and $O(w^2)$ processors~\cite{karp1991parallel}) 
and returns the $X_i$.
The algorithm proceeds to the next round if no such pair $(X_i,Y_i)$ is found.

To avoid generating an infinite sequence of random numbers at the beginning of the algorithm, we may apply the \emph{principle of deferred decisions}:
a random number is drawn only when it is accessed by the algorithm.
However,
it is crucial to ensure that the randomness used by the algorithm is {consistent}.
Specifically, the resolution of the same update $(v,i)$  in \Cref{alg:line:BP-update} of \Cref{alg:BP-simulation} must use the same random bits $\RandSeed_{(v,i)}$ every time.
The $\Sample$ subroutine is implemented using a dynamic data structure described in \Cref{alg:sampler-DS},
which generates random bits the first time they are needed and memorizes them for subsequent uses.

It is straightforward to verify that
\Cref{alg:sampler-DS} satisfies the correctness and coupling performance stated in \Cref{lemma:sample-correctness-efficiency} and  \Cref{lemma:sample-coupling-performance}. 
Additionally, according to \eqref{eq:sample-success-prob} and the independence between the pairs $(X_i,Y_i)$, the number of iterations of the \textbf{while} loop in \Cref{alg:sampler-DS} before termination follows a geometric distribution with success probability $1-(1-1/q)^{w}\ge 1-\mathrm{e}^{-w/q}$, where $q=|Q|$.

\begin{algorithm}[ht]
 \SetKwInOut{Input}{Input} \SetKwInOut{Output}{Output} 
\Input{ initial configuration $X_0\in Q^V$ and time $T>0$.}
\Output{ a random $X\in Q^V$ identically distributed as the $X_T$ in the continuous-time single-site dynamics $(X_t)_{t\in\mathbb{R}_{\ge 0}}$ specified by $\{P_v^{\tau}\}$ conditioning on $X_0$.}
\ForAllPar{$v\in V$}{
$T_0^v\gets0$\;
generate $0<T_1^v<\cdots<T_{M_v}^v<T$ by independently simulating a rate-1 Poisson process\;\label{alg:line-generate-times}
}
\ForAllPar{updates $(v,i)$ (where $v\in V$, $1\le i\le M_v$) and $u\in N_v^+$}{compute $\mathsf{pred}_u(v,i)\triangleq\max\{j\ge 0\mid T_j^u<T_i^v\}$\;\label{alg:line-compute-predecessor}
}
\ForAllPar{updates $(v,i)$ (where $v\in V$, $1\le i\le M_v$) }{
$\CurConfig{\widehat{X}}{0}{v}[i]\gets X_0(v)$\;
create a new instance of $\Sampler$  $\mathcal{S}_{(v,i)}$ and call $\mathcal{S}_{(v,i)}.\mathsf{Initialize}(\left\lceil |Q|\ln|V|\right\rceil)$\;\label{alg:line-sample-init}
}
$\ell\gets 0$\; 
\Repeat{$\widehat{X}^{(\ell)}=\widehat{X}^{(\ell-1)}$}{\label{alg:line:repeat-loop}
	$\ell\gets \ell+1$\;
	\lForAllPar{$v\in V$}{$\CurConfig{\widehat{X}}{\ell}{v}[0]\gets X_0(v)$}
	\ForAllPar{updates $(v,i)$ (where $v\in V$, $1\le i\le M_v$) }{
		let $\tau\in Q^{N_v^+}$ be constructed as: 
	$\forall u\in N_v^+$, $\tau_u\gets\CurConfig{\widehat{X}}{\ell-1}{u}[\mathsf{pred}_u(v,i)]$\; 
		$\CurConfig{\widehat{X}}{\ell}{v}[i]\gets \mathcal{S}_{(v,i)}.\mathsf{Draw}\left(P_v^{\tau}\right)$\; \label{alg:line:call-sample}
	}
} 
\Return{$X=\left(\CurConfig{X}{\ell}{v}[M_v]\right)_{v\in V}$}\;
\caption{$\mathsf{Simulation}\left(X_0,T\right)$}\label{alg:simulation}
\end{algorithm}

We now describe our parallel implementation of the algorithm that  faithfully simulates a continuous-time single-site dynamics $(X_t)_{t\in\mathbb{R}_{\ge 0}}$ up to an arbitrarily specified time $T>0$, given the initial configuration $X_0\in Q^V$.
The implementation is based on \Cref{alg:BP-simulation} and uses \Cref{alg:sampler-DS} as the $\Sample$ subroutine.
The random choices $(\mathfrak{T},\mathfrak{R})$ used by the chain are generated and accessed internally by the algorithm.
The detailed implementation is formally described in \Cref{alg:simulation}.

\subsection{Further Reducing the Total Work}

The number of processors used in \Cref{alg:simulation} grows at least linearly with $T$.
To reduce this cost, we apply a simple optimization:  
divide the time interval $[0,T]$ into smaller time intervals, each of length $O(\log n)$, and simulate the chain over these intervals sequentially.
This approach allows us to effectively treat $T$ as $O(\log n)$ when analyzing the processor and communication costs of the algorithm,  
while not increasing the number of parallel rounds.
The final algorithm, which uses \Cref{alg:simulation} as a subroutine, is described in \Cref{alg:efficient-simulation}.

\begin{algorithm}[ht]
 \SetKwInOut{Input}{Input} \SetKwInOut{Output}{Output} 
\Input{ initial configuration $X_0\in Q^V$ and time $T>0$.}
\Output{ a random $X\in Q^V$ identically distributed as the $X_T$ in the continuous-time single-site dynamics $(X_t)_{t\in\mathbb{R}_{\ge 0}}$ specified by $\{P_v^{\tau}\}$ conditioning on $X_0$.}
\While{$T>\ln(n)$}{
$X_0\gets\mathsf{Simulation}\left(X_0,\ln(n)\right)$\; 
$T\gets T-\ln(n)$\;
}
\Return{$\mathsf{Simulation}\left(X_0,T\right)$}\;
\caption{Parallel simulation of single-site dynamics with reduced total work}\label{alg:efficient-simulation}
\end{algorithm}

\begin{proof}[Proof of \Cref{thm:two-spin}]
It is straightforward to verify that \Cref{alg:simulation} simulates the process of \Cref{alg:BP-simulation} on the same initial configuration $X_0$ and the same random choices of $(\mathfrak{T},\mathfrak{R})$ used to generate the chain $(X_t)_{0\le t\le T}$. 
By \Cref{lemma:BP-convergence}, \Cref{alg:BP-simulation} returns the correct sample of $X_T$ upon termination, and so does \Cref{alg:simulation}.
The correctness of \Cref{alg:efficient-simulation} follows from the Markov property,
thus proving \Cref{itm:main-parallel-correctness} of \Cref{thm:two-spin}.

The output of \Cref{alg:sampler-DS} is identical to that of \Cref{alg:sample} on the same input distribution and random bits, 
ensuring that \Cref{alg:sampler-DS} is 2-competitive.
Under the assumption of \Cref{cond:main} with parameter $\rho$,
and choosing $T_0\le \ln(n)$,
\Cref{lemma:BP-fast} shows that
$\mathsf{Simulation}\left(X_0,T_0\right)$ terminates within $O(\rho\cdot T_0+\log\frac{n}{\epsilon})=O(\rho\cdot\log\frac{n}{\epsilon})$ iterations of the \textbf{repeat} loop in \Cref{alg:line:repeat-loop} of \Cref{alg:simulation} with probability at least $1-\epsilon$.
Within each iteration, the algorithm makes parallel calls to \Cref{alg:sampler-DS}.
With probability at least $1-\exp(-n)$ there are $O(nT_0)=O(n\ln(n))$ parallel calls to \Cref{alg:sampler-DS} in each iteration, 
because the total number of updates follows a Poisson distribution with mean $nT_0\le n\ln(n)$. 
For each individual call to \Cref{alg:sampler-DS}, the number of rounds required follows a geometric distribution with a success probability of at least $1-\mathrm{e}^{-\ln(n)}=1-1/n$.
This is bounded by $O(1+\log_n\frac{1}{\epsilon})$ with probability at least $1-\epsilon$, while each round takes $O(1)$ depth to compute.
Apart from the \textbf{repeat} loop, the depth of $\mathsf{Simulation}\left(X_0,T_0\right)$ is dominated by the sequential computations in \Cref{alg:line-generate-times} and \Cref{alg:line-compute-predecessor}, which are both bounded by $\max_{v\in V}M_v$, the maximum number of times a Poisson clock rings at a site. 
This is bounded by $O(\log\frac{n}{\epsilon})$ with probability at least $1-\epsilon$ due to the concentration of the Poisson random variable with mean $T_0\le \ln(n)$.
Overall, the depth of $\mathsf{Simulation}\left(X_0,T_0\right)$ for a $T_0\le \ln(n)$ is bounded by $O(\rho\cdot \log\frac{n}{\epsilon}(1+\log_n\frac{1}{\epsilon}))$ with probability at least $1-\epsilon-\exp(-\Omega(n))$.
Now consider \Cref{alg:efficient-simulation}, which runs for $\ln(n)<T\le n^{c_0}$. 
It divides $T$ into $T/\ln(n)$ parts. Choosing the error as $\epsilon= n^{-c_1}/T$, and applying the union bound, with probability at least $1-n^{-c_1}$, the depth of \Cref{alg:efficient-simulation} is bounded by $O_{c_1}(\frac{T}{\ln (n)}\cdot \rho \cdot \log(nT)\cdot \log_{n}(nT))=O_{c_0,c_1}(\rho\cdot T)$ because $T\le n^{c_0}$.
Taking the maximum of the $0<T\le \ln(n)$ and $T>\ln(n)$ cases, the depth of \Cref{alg:efficient-simulation} is always upper bounded by $O_{c_0,c_1}(\rho\cdot (T+\log n))$  with probability at least $1-n^{-c_1}$. 

\Cref{alg:efficient-simulation} uses the same number of processors as the $\mathsf{Simulation}\left(X_0,T_0\right)$ in \Cref{alg:simulation} for a $T_0\le \ln(n)$. 
It suffices to bound the number of processors used for the latter. 
The number of processors used by $\mathsf{Simulation}\left(X_0,T_0\right)$ is dominated by \Cref{alg:line-compute-predecessor} and the calls to \Cref{alg:sampler-DS} in \Cref{alg:line:call-sample}.
The former uses $\sum_{v\in V}M_v\mathrm{deg}(v)=\sum_{\{u,v\}\in E}(M_u+M_v)$ processors, where $M_v$, the number of times the Poisson clocks rings at site $v$, follows a Poisson distribution with mean $T_0\le\ln(n)$. 
Hence, the number of processors used in \Cref{alg:line-compute-predecessor} is bounded by $O_{c_1}(m\log n)$ with probability $\ge 1-n^{-c_1-1}$, where $m=|E|$. 
In each iteration, the number of parallel calls to \Cref{alg:sampler-DS} in \Cref{alg:line:call-sample} is given by the number of updates $M=\sum_{v\in V}M_v$, which follows a Poisson distribution with mean $nT_0\le n\ln(n)$, and is thus bounded by $O(nT_0)=O(n\log n)$ with probability $1-\exp(-\Omega(n))$.
Each call to \Cref{alg:sampler-DS} takes $(q\ln(n))^2$ processors for an $O(1)$ depth computation of the $\min$ index in the \CRAM{} model.
Overall, the total number of processors is bounded by $O_{c_1}((m+nq^2\log^2 n)\log n)$ with probability at least $1-n^{-c_1}$. 
This proves \Cref{itm:main-parallel-fast}.

Finally, to prove \Cref{itm:main-two-spin} of \Cref{thm:two-spin}, consider \Cref{alg:simulation}, 
where the $\Sample(\cdot,\cdot)$ subroutine is realized by inverse transform sampling on the Boolean domain, as described in \Cref{def:inverse-transform-sampling}.
Specifically, in \Cref{alg:simulation}, \Cref{alg:line-sample-init} is replaced by drawing $\RandSeed_{(v,i)}\in[0,1]$ uniformly and independently at random; and \Cref{alg:line:call-sample} is replaced by:
\[
\CurConfig{\widehat{X}}{\ell}{v}[i]\gets I[\RandSeed_{(v,i)}\ge P_v^{\tau}(0)].
\]
By \Cref{lemma:BP-fast-2spin}, $\mathsf{Simulation}\left(X_0,T\right)$ terminates within $O\left(\frac{1}{1-\rho}\log\left(\frac{nT}{\epsilon}\right)\right)$ iterations of the \textbf{repeat} loop in \Cref{alg:line:repeat-loop1} of \Cref{alg:simulation}, with probability at least $1-\epsilon$. 
The total number of updates follows a Poisson distribution with mean $nT$,
so with probability at least $1-\exp(-n)$, there are $O(nT)$ updates in total.
Each iteration of the \textbf{repeat} loop in \Cref{alg:line:repeat-loop1} of \Cref{alg:BP-simulation} uses $O(1)$ depth, taking $O_{c_1}(mT)$ processors. The $O_{c_1}(mT)$ processors can be reused between iterations. 
Before entering the \textbf{repeat} loop, the calculations of the predecessors of all updates in \Cref{alg:line-compute-predecessor} of \Cref{alg:simulation} require $O_{c_1}(mT)$ processors.  Overall, with probability at least $1-n^{-c_1}$, the depth of $\mathsf{Simulation}\left(X_0,T\right)$ is $ O_{c_1}(\frac{1}{1-\rho}(\log T+\log n))$ and it uses $O_{c_1}\left(mT\right)$ processors.
This proves \Cref{itm:main-two-spin}. 
\end{proof}

\begin{remark}[{bounded precision}]
The random choices in our algorithms---specifically,  the times generated by the Poisson clocks and the real number drawn uniformly from $[0,1]$ in the sampling subroutines \Cref{alg:sample} and \Cref{alg:sampler-DS}---are assumed to be real numbers of unbounded precision  for expositional simplicity.
However, this assumption can be relaxed in practical implementations.

First, the random update times $T_i^v$'s generated by Poisson clocks do not need to be real numbers.
They are used solely to determine the relation $\prec$  between updates, as defined in \eqref{order-def}.
The relation $\prec$ is uniquely determined as long as all times generated by the Poisson clocks are distinct.
This uniqueness is assured with high probability by generating only the first $O(\log n)$ bits for the times using standard methods for simulating Poisson point processes~\cite{stoyan1995stochastic}.

Furthermore, if all probabilities in the local update distributions $\{P_v^{\tau}\}$ of the single-site dynamics are expressed with $k$-bit precision (i.e.~they are rational numbers with a denominator of $2^{k}$),
we can construct the universal coupling (correlated sampling) only for those distributions $p\in\Delta(Q)$ specifically for those distributions with $k$-bit precision.
In \Cref{alg:sample}, \Cref{alg:sampler-DS} and \Cref{def:inverse-transform-sampling}, we can replace the real numbers drawn uniformly from $[0,1]$ with uniform $k$-bit rational numbers from $[0,1)$.
The analyses in \Cref{sec:universal-coupling} remain valid with this approach, as \eqref{eq:sample-success-prob} holds for any function with $k$-bit precision.
\end{remark}

\subsection{Distributed Implementation}

The following theorem is a formal restatement of \Cref{main-thm:distributed}.
\begin{theorem}[\Cref{main-thm:distributed}, formally stated]
Assume \Cref{cond:main} with parameter $\rho$.
There is a  \CONGEST{}  algorithm on the network $G=(V,E)$ such that, given any $X_0\in Q^V$ and  $0<T\le n^{c_0}$,  where $n=|V|$ and $q=|Q|$,
the followings hold:
\begin{enumerate}
\item\label{itm:main-distributed-fast}
The algorithm terminates within $O_{c_0,c_1}\left(\rho\cdot (T+\log n)\right)$ rounds of communication, with each message containing $O_{c_1}(\log n\log q)$ bits.
\item\label{itm:main-distributed-correctness}
The algorithm returns a random $X\in Q^V$ whose distribution is within $n^{-c_1}$ total variation distance from the distribution of $X_T$ in the continuous-time single-site dynamics $(X_t)_{t\in\mathbb{R}_{\ge 0}}$ specified by $\{P_v^{\tau}\}$ conditioned on~$X_0$.
\end{enumerate} 
\end{theorem}

\begin{algorithm}[ht]
 \SetKwInOut{Input}{Input at $v\in V$} \SetKwInOut{Output}{Output at $v\in V$} 
 \Input{ initial spin $X_0(v)$ of site $v\in V$ and $T>0$.}
 \Output{ the spin $X(v)$ of site $v$.}
 \SetKwBlock{PhaseI}{Phase I:}{end-of-phase}
  \SetKwBlock{PhaseII}{Phase II:}{end-of-phase}
 \PhaseI{
 generate $0<T_1^v<\cdots<T_{M_v}^v<T$ by independently simulating a rate-1 Poisson process\; 
 \lForAll{neighbors $u\in N_v$}{send $(T_i^v)_{1\le i\le M_v}$ to $u$ and receive $(T_i^u)_{1\le i\le M_u}$ from $u$}\tcc{in $O(T+\log n)$ rounds of communications w.h.p.}
 \lForAll{$1\le i\le M_v$ and $u\in N_v^+$}{compute $\mathsf{pred}_u(v,i)\gets\max\{j\mid T_j^u<T_i^v\}\cup\{0\}$\label{alg:line-distributed-predecessor}}\tcc{only needs the first $O(\log n)$ bits of each $T_j^u$ w.h.p.}
 }
 \SetKwFor{ForRound}{for rounds}{do}{end}
  \PhaseII{
  generate independent random $\RandSeed_{(v,i)}$ for all $1\le i\le M_v$, \hspace{300pt} 
  where $\RandSeed_{(v,i)}=\left(\left(X_1^{(v,i)},Y_1^{(v,i)}\right),\left(X_2^{(v,i)},Y_2^{(v,i)}\right),\ldots\right)$, each $\left(X_j^{(v,i)},Y_j^{(v,i)}\right)\in Q\times[0,1]$ chosen uniformly and independently at random\;
  $\CurConfig{\widehat{X}}{0}{v}[i]\gets X_0(v)$ for all $0\le i\le M_v$\;
  \ForRound{$\ell=1$ to $L$}{\tcc{for some properly fixed $L=O(\rho\cdot T+\log n)$}
  $\CurConfig{\widehat{X}}{\ell}{v}[0]\gets X_0(v)$\;
  retrieve  $\CurConfig{\widehat{X}}{\ell-1}{u}[j]$ for all $1\le j\le M_u$ from all neighbors $u\in N_v$\;\label{alg:line-distributed-retrieve}
  \tcc{$O((T+\log n)\log q)$-bits message from each neighbor w.h.p.}
  \ForAll{$1\le i\le M_v$}{
  let $\tau\in Q^{N_v^+}$ be constructed as: $\forall u\in N_v^+$, $\tau_u\gets\CurConfig{\widehat{X}}{\ell-1}{u}[\mathsf{pred}_u(v,i)]$\;
  $\CurConfig{\widehat{X}}{\ell}{v}[i]\gets \Sample\left(P_v^{\tau},\RandSeed_{(v,i)}\right)$\;\tcc{use the $\Sample$  in \Cref{alg:sample}}
  }
  }
  \Return{$X(v)=\CurConfig{\widehat{X}}{\ell}{v}[M_v]$}\;
  }
\caption{A \CONGEST{} algorithm for simulating single-site dynamics\label{alg:distributed-simulation}}
\end{algorithm}

\begin{proof}
A \CONGEST{}  algorithm on the network $G=(V,E)$ is described in \Cref{alg:distributed-simulation}.
The algorithm consists of two phases.
In \textbf{Phase I}, each site $v\in V$ locally generates all its update times $0<T_1^v<\cdots<T_{M_v}^v<T$ up to time $T$ and exchange them with all its neighbors.
After receiving all updates times from its neighbors, $v$ can locally determine the ``predecessors'' of all its updates $(v,i)$ in the relation  $\prec$ based on the update times within its neighborhood $N_v^+$.
This phase can be completed in $O(T+\log \frac{n}{\epsilon})$ rounds with probability at least $1-\epsilon$, as the number of updates $M_v$ for each node $v$ follows a Poisson distributions with a mean of $T$.
Each message contains an update time $T_j^u$, which does not need to be a real number with infinite precision. 
As long as the predecessors $\mathsf{pred}_u(v,i)$ in \Cref{alg:line-distributed-predecessor} are correctly computed, the algorithm will run correctly. 
And this occurs with probability at least $1-\epsilon$ using a message size of $O(\log\frac{n}{\epsilon})$ bits.
In \textbf{Phase II}, \Cref{alg:distributed-simulation} simulates the process of  \Cref{alg:BP-simulation}, with each round in \Cref{alg:distributed-simulation}  corresponding to an iteration of the \textbf{repeat} loop in \Cref{alg:BP-simulation}.
The only communication required is to retrieve the current neighborhood configuration,  as specified in \Cref{alg:line-distributed-retrieve}. 
For each pair of neighbors, this involves at most $O(T+\log\frac{n}{\epsilon})$ spins and can be represented using $O((T+\log\frac{n}{\epsilon})\log q)$ bits, with probability at least $1-\epsilon$, due to the concentration of Poisson random variable.
According to \Cref{lemma:BP-fast}, with probability at least $1-\epsilon$, \textbf{Phase II} requires at most $O(\rho\cdot T+\log\frac{n}{\epsilon})$ rounds  to reach  a fixpoint. 
At this point, the algorithm achieves the $O_{c_1}(\rho\cdot T+\log n)$ round complexity  with probability at least $1-n^{-c_1}$, as stated in the theorem, although each message consists of $O_{c_1}((T+\log n)\log q)$ bits. 

To further reduce the message size to $O(\log n\log q)$ bits, we apply the same trick as in \Cref{alg:efficient-simulation}. 
For $\ln(n)<T\le n^{c_0}$, we divide $T$ into $T/\ln(n)$ many subintervals of length $T_0\le\ln (n)$ and run  \Cref{alg:distributed-simulation} consecutively for these time subintervals, such that \Cref{alg:distributed-simulation} is executed within each time subinterval for a properly fixed number of rounds without global coordination.
This approach works as long as no errors occur. 
By setting the error probability to  $\epsilon=n^{-c_1}/T$ and applying the union bound,  with probability at least $1-n^{-c_1}$, no error occurs, and the output is correct.  
 In the event of an error,  the output is treated as arbitrary, introducing a total variation distance error of $n^{-c_1}$ to the final result.
With probability at least $1-n^{-c_1}$, the algorithm that divides $T$ into subintervals and runs \Cref{alg:distributed-simulation} as subroutines terminates in $O_{c_1}(\frac{T}{\ln n}\cdot(\rho\cdot \ln n+\log (nT)))=O_{c_0, c_1}(\rho\cdot T)$ rounds since $T\le n^{c_0}$, with each message consisting of $O_{c_1}(\log n\log q)$ bits. 
In the case where $T<\ln(n)$,  with probability $\ge 1-n^{-c_1}$, \Cref{alg:distributed-simulation} terminates in $O_{c_1}(\rho\cdot T+\log n)=O_{c_1}(\rho \cdot \log n)$ rounds, with each message consisting of $O_{c_1}((T+\log n)\log q)=O_{c_1}(\log n\log q)$ bits.

Overall, with probability at least $1-n^{-c_1}$, the final algorithm (which applies \Cref{alg:distributed-simulation}  with the trick as in \Cref{alg:efficient-simulation}) terminates within $O_{c_0,c_1}(\rho\cdot (T+\log n))$ rounds, with each message consisting of $O_{c_1}(\log n\log q)$ bits.
The algorithm faithfully simulates \Cref{alg:BP-simulation}  except for an error probability of at most $n^{-c_1}$.
\end{proof}


\section{Applications}\label{sec:applications}
\subsection{Application for Hardcore and Ising Model}
\begin{proof}[Proof of \Cref{cor:Ising}]
For the hardcore models with $\lambda\le(1-\delta)\lambda_c(\Delta)=(1-\delta)\frac{(\Delta-1)^{\Delta-1}}{(\Delta-2)^{\Delta}}$,
and Ising models with $\beta\in\left(\frac{\Delta-2+\delta}{\Delta-\delta},\frac{\Delta-\delta}{\Delta-2+\delta}\right)$,
the modified log-Sobolev inequalities (MSLI) established in \cite{chen2022optimal} show that
the continuous-time Glauber dynamics has mixing time  $t^{\mathsf{C}}_{\mathsf{mix}}(\epsilon)=O_\delta\left(\log \frac{n}{\epsilon}\right)$.
For the hardcore models with $\lambda=O\left(\frac{1}{\Delta}\right)$ and Ising models with $\beta\in1\pm O\left(\frac{1}{\Delta}\right)$, 
the Dobrushin's influence matrix for the Glauber dynamics satisfies $\Inf{u}{v}=0$ for non-adjacent $u,v\in V$, 
and satisfies the following for adjacent $\{u,v\}\in E$, where the degree of $v$ is denoted as $d_v\triangleq |N_v| \leq \Delta$:
\begin{align*}
    &\text{(hardcore model)} &\Inf{u}{v} &= \frac{\lambda}{1+\lambda} = O\left(\frac{1}{\Delta}\right),\\
    &\text{(Ising model)} &\Inf{u}{v} &= \min_{0\le k\le d_v-1} \left| \frac{1}{1+\lambda \beta^{d_v-2k}}-\frac{1}{1+\lambda \beta^{d_v-2(k+1)}} \right| = O\left(\frac{1}{\Delta}\right).
\end{align*}
Thus, we have $\|\InfMat\|_\infty=O(1)$,
meaning \Cref{cond:main} holds with parameter $\rho=O(1)$.

Applying \Cref{thm:parallel} with $T=t^{\mathsf{C}}_{\mathsf{mix}}\left(\frac{1}{\mathrm{poly}(n)}\right)=O_\delta(\log n)$,   \Cref{alg:efficient-simulation} returns an approximate sample that is $1/{\mathrm{poly}(n)}$-close in total variation distance to the Gibbs distribution. This algorithm uses $O_{\delta}(\log n)$ depth on $\tilde{O}(m+n)=\tilde{O}(m)$ processors with high probability, assuming the availability of oracles to evaluate the marginal probabilities:
\begin{align*}
&\text{(hardcore model)} & \forall \tau\in\{0,1\}^{N_v},  &\quad\mu_v^\tau(1) = 
\begin{cases}
    0 & \|\tau\|_1 \ge 1\\
    \frac{\lambda}{1+\lambda} & \|\tau\|_1 = 0 
\end{cases}, \\
&\text{(Ising model)} & \forall \tau\in\{0,1\}^{N_v},  &\quad \mu_v^\tau(1) = \frac{\lambda\beta^{-{d}_v+2\|\tau\|_1}}{{1+\lambda\beta^{\|\tau\|_1/({d}_v-\|\tau\|_1)}}}.
\end{align*}
These oracles can be implemented with  $O(\log d_v)=O(\log \Delta)$ depth using $O(d_v)$ processors. 
The overhead in processor usage contributes only to the $\tilde{O}(n)$ term in the $\tilde{O}(m+n)$  total processor bound.
This results in an overall bound of $O_{\delta}(\log n \cdot \log \Delta)$ depth and $\tilde{O}(m+\sum_{v\in V}d_v)=\tilde{O}(m)$ processors.
\end{proof}

\subsection{Application for SAT Solutions}
\begin{proof}[Proof of \Cref{cor:CNF}]
Let $\mu$ denote the uniform distribution over all satisfying assignments of $\Phi$.

The sampling algorithm proposed in \cite{feng2021fast} is composed of three key steps:
\begin{enumerate}
\item 
Constructing the set $\mathcal{M} \subseteq V$ of ``marked'' variables:
A set $\mathcal{M} \subseteq V$ that satisfies \cite[Condition 3.1]{feng2021fast} is constructed using the Moser-Tardos algorithm. 
This can be parallelized via the parallel Moser-Tardos algorithm \cite{moser2010constructive}. The local lemma condition \eqref{eq:LLL-cond} is sufficiently strong to allow the construction of $\mathcal{M} \subseteq V$ in parallel while ensuring \cite[Condition 3.1]{feng2021fast}  is met.
\item
Simulate the Glauber dynamics for  $\mu_\mathcal{M}$:
The Glauber dynamics is then simulated for the distribution $\mu_\mathcal{M}$, which is the projection of $\mu$ onto $\mathcal{M}$.
This involves a single-site dynamics on the complete graph induced by  $\mathcal{M}$, 
with local update distributions $\{P_v^\tau\}$ defined as $P_v^\tau=\mu_v^{\tau_{\mathcal{M} \setminus \{v\}}}$, where $\tau$ is the current assignment on $\mathcal{M}$.
According to \cite[Lemma 4.2]{feng2021fast}, all configurations on $\mathcal{M}$ have positive measure under $\mu_\mathcal{M}$, 
ensuring that $\{P_v^\tau\}$ are well-defined. 
By \cite[Lemma 4.1]{feng2021fast}, this chain mixes within $t_{\mathsf{mix}}^{\mathsf{D}}(\epsilon)=O(n\log \frac{n}{\epsilon})$ steps, which translates to a continuous-time mixing time $t_{\mathsf{mix}}^{\mathsf{C}}(\epsilon)=O(\log \frac{n}{\epsilon})$  via \eqref{eq:continuous-discrete-mixing}.
This mixing time bound arises from a path coupling argument   \cite[Equation (11)]{feng2021fast}, 
which also ensures \Cref{cond:main} holds with parameter $\rho=O(1)$ in the $\ell_\infty$ norm.
Applying \Cref{thm:parallel}  with $T=t_{\mathsf{mix}}^{\mathsf{C}}\left(\frac{1}{\mathrm{poly}(n)}\right)=O(\log n)$, \Cref{alg:efficient-simulation} simulates this chain within $O(\log n)$ depth  using $\tilde{O}(|\mathcal{M}|^2)=\tilde{O}(n^2)$ processors, given access to oracles for evaluating $\mu_v^{\tau_{\mathcal{M} \setminus \{v\}}}$.

To draw from the marginal distribution $\mu_v^{\tau_{\mathcal{M} \setminus \{v\}}}$, \cite[Algorithm 3]{feng2021fast} is used.
During each update of the chain, with high probability, the clauses already satisfied by the current assignment on  ${\mathcal{M} \setminus \{v\}}$ disconnect the CNF into components of small sizes. 
Rejection sampling can be applied to the component containing $v$, with each trial succeeding with a probability of at least $n^{-2^{-20}}$.

\item Extending the partial assignment of ``marked'' variables to all variables:
After the Glauber dynamics on $\mu_\mathcal{M}$ has sufficiently mixed, the random assignment on $\mathcal{M}$ is extended to all variables.
This extension is also performed using \cite[Algorithm 3]{feng2021fast}, which samples from $\mu_{V\setminus\mathcal{M}}^{\tau_{\mathcal{M}}}$.
\end{enumerate}
It remains to ensure the parallelizability of: (1) The marginal sampler \cite[Algorithm 3]{feng2021fast}, and (2) the oracles to evaluate  the marginal distribution $\mu_v^{\tau_{\mathcal{M} \setminus \{v\}}}$.
For (1),  the marginal sampler \cite[Algorithm 3]{feng2021fast} draws from $\mu_S^{\tau_T}$, where $S,T\subset V$ are disjoint. 
It first constructs the connected components that intersect $S$ in the formula of clauses not yet satisfied by ${\tau_T}$,
and then draws from $\mu_S^{\tau_T}$ via rejection sampling on these components. 
The first step can be parallelized using a connected component algorithm such as \cite{shiloach1980log} and the latter is trivially parallelizable, and the subsequent rejection sampling is trivially parallelizable.
For (2), due to the Chernoff-Hoeffding bound, the marginal probabilities $\mu_v^{\tau_{\mathcal{M} \setminus \{v\}}}$ can be estimated with an additive error of $1/\mathrm{poly}(n)$ with $1-1/\mathrm{poly}(n)$ confidence by independently (and thus in parallel) repeating \cite[Algorithm 3]{feng2021fast}  $\mathrm{poly}(n)$  times. 
This introduces only $1/\mathrm{poly}(n)$ total variation error into the final output. 
\end{proof}


\section{Conclusion and Open Problems}
In this work, we present a generic parallel algorithm for faithful simulation of single-site dynamics.
Assuming a substantially weakened asymptotic variant of $\ell_p$-Dobrushin's condition,
our algorithm achieves linear speedup in  $n$
when parallelizing single-site dynamics across  $n$ sites.
If the strict $\ell_p$-Dobrushin's condition is assumed for Boolean random variables,
the parallelization achieves an exponential speedup.

The asymptotic Dobrushin's condition required for linear speedup is relatively easy to satisfy.
It essentially  requires  that the discrepancy in the system does not propagate at a super-exponential rate,
which is quite modest compared to the usual requirement for fast mixing, where discrepancy should decay.

Under this mild condition,
our parallel simulation algorithm can convert single-site dynamics with near-linear mixing time into \RNC{} algorithms for sampling, provided the marginal distributions for single-site updates are \RNC{}-computable.
Additionally, through non-adaptive simulated annealing, we can also obtain \RNC{} algorithms for approximate counting.

\paragraph{Open problems.}
The study of efficient parallel algorithms for sampling and counting is highly motivated by practical applications. 
In this paper, we address these challenges through the faithful parallel simulation of single-site dynamics and non-adaptive simulated annealing. 
Our approach is generic and well-suited for leveraging existing results on single-site Markov chains with well-studied mixing properties.
Despite the advancements provided by our approach, 
the problem of giving \RNC{} counterparts for well-known sampling and counting algorithms remains widely open.
We conclude by presenting a few concrete open problems:
\begin{itemize}
\item
We apply our result to the chain in \cite{feng2021fast}, which provides a parallel sampler for satisfying solutions of the CNFs in the local lemma regime. However, the total work of this sampler is a large polynomial. A key open question is to design a parallel sampler with polylogarithmic depth while maintaining total work comparable to the sequential sampler in \cite{feng2021fast}, which is close to linear in $n$, the number of variables.
\item
In this paper, we utilize correlated sampling to ensure both the efficiency and faithfulness of parallel simulations for single-site dynamics. To our knowledge, this marks the first application of correlated sampling in parallelizing stochastic processes. An intriguing direction for future research is how this approach can be generalized to parallelize randomized algorithms.
\item 
A significant challenge is to parallelize Markov chains with super-linear polynomial mixing times into \RNC{} algorithms.
Notable examples include the Jerrum-Sinclair chains for matchings and the ferromagnetic Ising model with no external field~\cite{jerrum2003counting}.
Finding \RNC{} counterparts for these algorithms  remains a major open problem.
\item
An even greater challenge is to give an \NC{} (deterministic) algorithm for approximating counting, 
 especially for graphical models that have unbounded maximum degree. 
\end{itemize}



\newcommand{\etalchar}[1]{$^{#1}$}

\end{document}